\documentclass[preprint,12pt]{elsarticle}
\usepackage{amsfonts, amsmath, amssymb, amsthm, alltt, color, pbox, enumitem}
\usepackage{algorithm}
\usepackage{algcompatible}

\usepackage{ftnxtra}
\usepackage{afterpage}

\newtheorem{theorem}{Theorem}[section]
\newtheorem{lemma}[theorem]{Lemma}
\newtheorem{proposition}[theorem]{Proposition}
\newtheorem{corollary}[theorem]{Corollary}
\newtheorem{postulate}{Postulate}
\theoremstyle{definition}
\newtheorem{definition}{Definition}[section]
 \newtheorem{example}{Example}[section]
\theoremstyle{remark}
\newtheorem{remark}{Remark}

\numberwithin{figure}{section}

\newcommand{\setofstates}[2]{\mathcal{#1}_#2}
\newcommand{\multiset}[1]{\{\!\!\{{#1}\}\!\!\}}

\renewcommand{\thefootnote}{\fnsymbol{footnote}}

\usepackage{xcolor}
\newif\ifshowchanges
\newcommand{\new}[1]{{\ifshowchanges\color{blue}\fi#1}}
\newenvironment{newtext}{\ifshowchanges\color{blue}\fi}{}
\newcommand{\removed}[1]{\ifshowchanges{\color{red}[\textit{Removed:} #1]}\else\expandafter\ignorespaces\fi}

\newcommand{\val}{\mathrm{val}}
\newcommand{\str}{\mathit{str}}
\newcommand{\FOwo}{\mathrm{FO}_{wo=}}

\journal{Theoretical Computer Science}

\begin{document}
 
\begin{frontmatter}



\title{A Behavioural Thesis of Synchronous Parallel Algorithms -- Correction of the Simplified Parallel ASM Thesis\footnote{The research reported in this paper results
    from the project \textit{Behavioural Theory and Logics for
      Distributed Adaptive Systems} supported by the \textbf{Austrian
      Science Fund (FWF): [P26452-N15]}.}}


\author[1]{Flavio Ferrarotti}
\author[2]{Klaus-Dieter Schewe}
\author[2]{Loredana Tec}
\author[3]{$\dagger$Qing Wang}
 
\address[1]{Software Competence Center Hagenberg, Austria, \textrm{flavio.ferrarotti@scch.at}}
\address[2]{Linz, Austria, \textrm{\{kd.schewe | loredana.tec\}@gmail.com}}
\address[3]{Research School of Computer Science, The Australian National University, Australia}

\begin{abstract}

A behavioural theory consists of machine-independent postulates characterizing a particular class of algorithms or systems, an abstract machine model that provably satisfies these postulates, and a rigorous proof that any algorithm or system stipulated by the postulates is captured by the abstract machine model. The class of interest in this article is that of synchronous parallel algorithms. For this class a behavioural theory has already been developed by Blass and Gurevich, which unfortunately, though mathematically correct, fails to be convincing, as it is not intuitively clear that the postulates really capture the essence of (synchronous) parallel algorithms. 

In this article we present a much simpler (and presumably more convincing) set of four postulates for (synchronous) parallel algorithms, which are rather close to those used in Gurevich's celebrated sequential ASM thesis, i.e. the behavioural theory of sequential algorithms. The key difference is made by an extension of the bounded exploration postulate using multiset comprehension terms instead of ground terms formulated over the signature of the states. In addition, all implicit assumptions are made explicit, which amounts to considering states of a parallel algorithm to be represented by meta-finite first-order structures. 

The article first provides the necessary evidence that the axiomatization presented in this article characterizes indeed the whole class of deterministic, synchronous, parallel algorithms, then formally proves that parallel algorithms are captured by Abstract State Machines (ASMs). 
\new{The proof uses the requirement of the bounded exploration postulate that update sets be uniformly definable, by equality-free first-order formulae, in the finite critical structures determined by the witness terms; these formulae are translated into the guards of a parallel ASM rule. This version corrects the characterization proof of the first version, whose key lemma was false; the bounded exploration postulate has been amended accordingly.}

\end{abstract}


\end{frontmatter}

\renewcommand{\thefootnote}{\arabic{footnote}}
\setcounter{footnote}{0}

\begin{newtext}
\paragraph*{Note on this version}
This version corrects the first arXiv version (v1, April 2015) of this article; a revised form of v1 appeared as~\cite{[FSTW16]}. Lemma~7.11 of v1 is false, and the characterization theorem of v1 does not hold under the bounded exploration postulate as stated there: some objects satisfying the postulates of v1 are not behaviourally equivalent to any parallel ASM over a computable background. The corrections follow~\cite{[FS-Corr]}:
\begin{itemize}
\item the Bounded Exploration Postulate (Postulate~\ref{boundedexpl}) additionally requires that update sets be uniformly definable in the critical structures, which are now introduced in Section~\ref{sec:boundedExploration};
\item the bounded exploration witness of a parallel ASM (Definition~\ref{def:witnessSet}) and the proof of the plausibility theorem (Theorem~\ref{Th:Plausibility}) are replaced;
\item in Section~\ref{sec:examples}, the claims that the listed witness terms form bounded exploration witnesses are replaced by a reference to Theorem~\ref{Th:Plausibility};
\item the proof of the characterization theorem (Section~\ref{sec:characterizationthm}) is replaced.
\end{itemize}
The abstract, the introduction and the conclusions are adjusted accordingly. All other parts are unchanged.\ifshowchanges{} In this copy, new text is shown in blue and removed text is indicated in red.\fi
\end{newtext}


\section{Introduction}

The starting point for this research was the attempt of Gurevich to characterize algorithms by means of Abstract State Machines (ASMs). The so-called ASM thesis, first proposed in 1985 in a note to the American Mathematical Society~\cite{[Gurevich85]}, asserts that every algorithm is equivalent, on its natural level of abstraction, to an appropriate abstract state machine. In~\cite{[Gurevich00]}, Gurevich formulated and proved the ASM thesis for sequential algorithms. This consists of three intuitive postulates (sequential time, abstract state, bounded exploration) that are used to define sequential algorithms at any level of abstraction, and a proof that algorithms defined in this way are exactly captured by sequential ASMs. Starting from the sequential ASM thesis, a set of postulates for (synchronous) parallel algorithms has been proposed by Blass and Gurevich in~\cite{[BG00],[BG03]}. However, these postulates turned up being significantly more complex and less intuitive than the ones for the sequential setting and thus have not fully convinced the ASM community nor others dealing with foundations of computing and rigorous methods.

Our intention in this article is to formulate and prove a thesis for parallel algorithms that is similar to Gurevich's ASM thesis for sequential algorithms and thus overcomes the lack of intuition in the work of Blass and Gurevich\footnote{In the end it will turn out that both sets of postulates are exactly captured by ASMs, which implies their equivalence.}.
The key idea is to relax the bounded exploration postulate from the sequential setting by using more general multiset comprehension terms instead of only ground terms. In doing so we present a new parallel ASM thesis for synchronous, parallel algorithms, which consists of: (a) four postulates that capture the fundamental properties of synchronous parallel algorithms (\emph{axiomatization}); (b) a variant of Abstract State Machines (which we call \emph{parallel} ASMs) together with a proof that every parallel ASM satisfies the postulates (\emph{plausibility theorem}); and (c) a proof that every algorithm stipulated by the postulates can be step-by-step simulated by an equivalent parallel ASM (\emph{characterization theorem}).

There are three postulates that are common to the sequential ASM thesis of Gurevich \cite{[Gurevich00]} and to the parallel ASM thesis of Blass and Gurevich \cite{[BG03],[BG08]}: the \emph{Sequential Time}, \emph{Abstract state} and \emph{Background} postulates. While the first two postulates are identical in both theses, the third postulate (although only implicit in the sequential thesis) states the differences between sequential and parallel algorithms in the background of computations. This is because the minimum background required for the computation of parallel algorithms is nonetheless bigger than the minimum background required for sequential algorithms. 

The last postulate in the sequential ASM thesis is the \emph{Bounded Exploration} postulate, which basically says that every sequential algorithm examines only a bounded number of elements in any state, with the number of elements to be examined being bounded uniformly by the algorithm and not by the state. Unfortunately, the three postulates that replace it in the parallel ASM thesis, called the \emph{Proclet Algorithm}, \emph{Bounded Sequentiality} and \emph{Update} postulates, are not as concise and intuitive. They are based on a number of non-trivial concepts such as proclet, ken, information flow digraph, etc., and thus difficult to explain in a concise and intuitive manner unless previous knowledge of such concepts is assumed\footnote{In fact, different to bounded exploration these concepts are not purely grounded in logic.}.
To gain a better understanding of the properties of parallel algorithms, we propose to replace these three postulates in the parallel ASM thesis by a generalized Bounded Exploration postulate that is tailored to parallel algorithms. This new postulate is based on the observation that only finitely many locations of a state can be changed in one step of a computation, independently of whether the algorithm is sequential or parallel.

By the sequential accessibility principle in \cite{[Gurevich00]}, the only way in which a sequential algorithm can access an element $a$ of a state, is by producing a ground term that evaluates to $a$. This principle, together with the informal assumption that every algorithm has a finite program, indicates that sequential algorithms can only check agreement between states on a \emph{fixed} finite set of elements denoted by a fixed finite set of ground terms (the witness set in the Bounded Exploration postulate for sequential algorithms). Parallel algorithms, however, do not satisfy this principle (see Example~\ref{Ex:GraphComplement} extracted from \cite{[BG03]}).

\begin{example}
\label{Ex:GraphComplement}

\ The following algorithm takes as input an undirected graph and transforms it to its complement. We assume that the vocabulary of the states of this algorithm includes the function symbols $V$ and $E$ which are interpreted as the vertex-set and edge-set of the graph, respectively.

\begin{alltt}
forall \(x, y\) with \(V(x) \wedge V(y)\) do
   if \(x \neq y \) then \(E(x,y)\) := \(\neg(E(x,y))\) endif
enddo
\end{alltt}

Clearly, the produced set of updates to the edge set that interprets $E$ depends on the entire graph while two graphs might be very different despite any amount of \emph{fixed} finite agreement between them.
Thus, the hypothesis in the bounded exploration postulate for sequential algorithms when applied to this parallel algorithm cannot guarantee any agreement at all between two arbitrary graphs.

\end{example}

Intuitively, in a parallel algorithm many ``branches'' contribute to the update set produced during a single computation step. Although still finite, the actual number of branches is no longer bounded by the algorithm alone, but also depends on the current state of the algorithm. This number is nevertheless ``uniformly'' determined by the algorithm. This motivated us to think about an alternative Bounded Exploration postulate for parallel algorithms based on a more expressive set of terms. The new formulation of the Bounded Exploration postulate allows us to ``capture'' in every state $\bf S$ of a parallel algorithm $A$, the part of $\bf S$ that is actually explored by the branches of $A$. \new{In addition, the update set must be determined by this explored part in a uniform way, i.e., by a fixed first-order description that does not depend on the state. For sequential algorithms this uniformity is a consequence of bounded exploration; for parallel algorithms it has to be required explicitly (see Section~\ref{sec:boundedExploration}).}

\subsection{Related Work}
\label{Sec:stateofart}

The seminal work of Gurevich~\cite{[Gurevich95]} on Abstract State Machines (ASMs, formerly called ``evolving algebras'') aimed to find a precise formal definition of the notion of algorithm. The major discovery was that all computing formalisms were bound to a specific abstraction level, which implied that almost always encodings were required \cite{[Gurevich00],[Gurevich2012]}, so the major breakthrough of ASMs was due to the abstract notion of state, which is defined by general Tarski structures. 

The sequential ASM thesis~\cite{[Gurevich00]} characterizes sequential algorithms in terms of three postulates: sequential time, abstract state and bounded exploration. Moreover, it  establishes the characterization theorem for sequential algorithms stating that algorithms defined this way are exactly captured by sequential ASMs, i.e. a well-defined abstract machine model that relies on the parallel execution of updates on the abstract states, if certain conditions guarding the updates are satisfied. Thus, also sequential ASMs support bounded parallelism, where the bound is a priori fixed by the algorithm and does not depend on the state.

Following this work, many other extensions of the sequential ASM thesis to different classes of algorithms have been explored. These include the parallel ASM thesis \cite{[BG03],[BG08]} for parallel algorithms, in which the bound on the parallel branches in a computation is dropped. The postulates are significantly more complex, as multisets and multiset operations must be provided explicitly in the background \cite{[BG07c]} to permit branching and synchronization. While sequential time and abstract state postulates are preserved, the bounded exploration postulate has been replaced by a set of postulates that permit to distinguish between local and global states. There is still a debate in the ASM community, if these postulates can be simplified to obtain an axiomatization for parallel algorithms that is as intuitive as the one for sequential algorithms. A simplified parallel ASM thesis has been conjectured in~\cite{[SW2012a]}.

The approach used in the sequential ASM thesis has also been successfully adopted in the development of a theory of sequential database transformations~\cite{[SW2010]}, where the key problem is to cope with the intrinsic finiteness of databases and the need to capture the constructs and operators that are defined by a data model. This was solved by adopting meta-finite structures~\cite{[Graedel98]} for the states, and introducing explicitly background structures that capture the necessary constructs for the data models, e.g. trees, hedges and hedge algebra operations in the case of XML \cite{[SW2010a]}. This has been extended to synchronous, parallel database transformations in~\cite{[SW2012]}, which has given hints towards a simplification of the parallel ASM thesis. A conjecture concerning a modified (Parallel) Bounded Exploration postulate was already formulated in~\cite{[SW2012]}.

\subsection{Outline}

The remainder of this article is organized as follows. We begin with some preliminaries in Section~\ref{Sec:Prelim}, which include also the first two postulates, which are common for both sequential and parallel settings. Then in Section~\ref{sec:backgroundOfComputation} we highlight implicit assumptions regarding states and background of computation. This leads us to consider states as meta-finite structures and to technically distinguish states of an algorithm (as introduced in the Abstract State postulate) from states of computation (that contain many extensions as per the Background postulate). The use of meta-finite structures to represent states permits to capture in a natural way the finite components of the states, making explicit their intrinsic finiteness. In particular, the technicality of assuming that every state includes a finite set of processes or branches (e.g., proclets in~\cite{[BG03],[BG08]}), is not longer required.

After introducing a bounded exploration postulate for parallel algorithms in Section~\ref{sec:boundedExploration}, we discuss in Section~\ref{sec:asmdef} the plausibility of an ASM thesis for the class of synchronous parallel algorithms defined by the set of postulates proposed in this work. That is, we define an ASM model which captures the class of algorithms satisfying the Sequential Time, Abstract State, Background and (Parallel) Bounded Exploration postulates as defined in Sections~\ref{Sec:Prelim}--\ref{sec:boundedExploration}. Our formal ASM model is similar to the ASM model used in the parallel ASM thesis~\cite{[BG03],[BG08]} modulo some technical details which do not affect the expressiveness of the model. These technical differences are due to the fact that we represent states as meta-finite structures instead of assuming that the states include a finite set of proclets.  

In Section~\ref{sec:examples}, we illustrate  how the proposed set of postulates characterizes powerful models of parallel computation such as parallel random access machines, circuits, alternating Turing machines, first order logic and MapReduce. The main result, which is the characterization theorem, stating that for every parallel algorithm there exists an equivalent parallel ASM, is presented in detail in Section~\ref{sec:characterizationthm}.

\section{Preliminaries}
\label{Sec:Prelim}

Following Gurevich's {ASM} thesis for sequential algorithms \cite{[Gurevich00]} and Blass and Gurevich's {ASM} thesis for parallel algorithms \cite{[BG03],[BG08]}, our first postulate states that a parallel algorithm is a synchronous algorithm, i.e., that in every possible computation there is an initial state, followed by a second state, followed by a third state and so on, and that the progression from one state to the next is uniquely determined by the algorithm. 

\begin{postulate}[Sequential Time Postulate]
\label{sequentialTime}\rm
A \emph{parallel algorithm} $A$ is associated with a non-empty set $\mathcal{S}_A$ of \emph{states}, a non-empty subset $\mathcal{I}_A \subseteq \mathcal{S}_A$ of \emph{initial states}, and a function $\tau_A : \setofstates{S}{A} \rightarrow \setofstates{S}{A}$ called the \emph{one-step transformation} of $A$.
\end{postulate}

A \emph{run} or a \emph{computation} of  a parallel algorithm $A$ is a finite or infinite sequence of states ${\bf S}_0, {\bf S}_1, \ldots$, where ${\bf S}_0$ is an initial state in $\mathcal{I}_A$ and ${\bf S}_{i+1} = \tau_A({\bf S}_i)$ holds for every $i \geq 1$. A state ${\bf S}$ of $A$ is called \emph{reachable} if ${\bf S}$ occurs in some run of $A$.

We use the following (strong) notion of equivalence among algorithms, which implies that behaviourally equivalent parallel algorithms have the same runs. 
\begin{definition}[Behavioural Equivalence]
\label{BEquiv}
Algorithms $A$ and $B$ are said to be \emph{behaviourally equivalent} if $\setofstates{S}{A} = \setofstates{S}{B}$, $\setofstates{I}{A} = \setofstates{I}{B}$ and $\tau_A = \tau_B$ hold. 
\end{definition}

We follow Gurevich's approach in which states are full instantaneous descriptions of the algorithm that can be conveniently formalized as first-order structures. More formally, we consider states as first-order structures whose \emph{vocabulary} or \emph{signature} $\Sigma$ is a finite set of function symbols. Each function symbol $f_i \in \Sigma$ has a fixed \emph{arity} $r_i \geq 0$. Function symbols can be marked (by the vocabulary) as \emph{static}. Otherwise, they are \emph{dynamic}. A \emph{first-order structure} ${\bf S}$ of vocabulary $\Sigma$ is a nonempty set $S$ called the \emph{base set} of ${\bf S}$ together with interpretations of every function symbol in $\Sigma$ over $S$. Elements of $S$ are also called elements of the structure $\bf S$. An \emph{interpretation} of an $r$-ary function symbol $f \in \Sigma$ over $S$ is a (total) function $f^{\bf S}$ from $S^r$ to $S$. Consequently, our second postulate is also unchanged from previous work in the area \cite{[BG03],[BG08],[Gurevich00]}.

\begin{postulate}[Abstract State Postulate]
\label{abstractState}\rm

States of a parallel algorithm $A$ are first-order structures. All states in $\mathcal{S}_A$ have the same vocabulary. The one-step transformation $\tau_A$ does not change the base set of any state. $\setofstates{S}{A}$ and $\setofstates{I}{A}$ are closed under isomorphisms. Any isomorphism between two states ${\bf S}_1$ and ${\bf S}_2$ is also an isomorphism between $\tau_A({\bf S}_1)$ and $\tau_A({\bf S}_2)$.
\end{postulate}

Apart from defining states of an algorithm as first-order structures of a fixed vocabulary, this postulate ensures that algorithms work at a fixed level of abstraction by requiring the set of states of an algorithm to be closed under isomorphisms and the one-step transformation to preserve those isomorphisms. Recall that two structures ${\bf S}_1$ and ${\bf S}_2$ of a same vocabulary $\Sigma$ are \emph{isomorphic} (denoted ${\bf S}_1 \simeq {\bf S}_2$) iff there is a bijection $\zeta : S_1 \rightarrow S_2$ between the base sets such that $\zeta(f^{{\bf S}_1}(a_1, \ldots, a_r)) = f^{{\bf S}_2}(\zeta(a_1), \ldots, \zeta(a_r))$ holds for all $r$-ary function symbol $f \in \Sigma$ and all $r$-tuples $(a_1,\ldots, a_r) \in S_1^r$.  

As usual, we think of a structure that represents a state of an algorithm as a kind of memory that maps locations to values.

\begin{definition}[Locations and Updates]
Let $\bf S$ be the state of an algorithm of vocabulary $\Sigma$, let $f \in \Sigma$ be a function symbol of arity $r$ and let $\bar{a}$ be an $r$-tuple in $S^r$. The pair $(f,\bar{a})$ represents a (memory) \emph{location} in ${\bf S}$. The \emph{content} of the location $(f,\bar{a})$ is the value $f^{\bf S}(\bar{a})$ in $S$. If $l = (f,\bar{a})$ is a location in a state ${\bf S}$, $f$ is a dynamic function and $b$ is an element in the base set $S$, then the tuple $(l, b)$ is an \emph{update} of ${\bf S}$. If $b = f^{\bf S}(\bar{a})$, then the update $((f, \bar{a}), b)$ is called a \emph{trivial update}.
\end{definition}

An update $(l,b)$ indicates that the content of the location $l$ in $\bf S$ needs to be changed to the value $b$. Two updates \emph{clash} if they refer to the same location but are distinct.

\begin{definition}[Consistent Update Set]
A set of updates $\Delta$ is \emph{consistent} if it has no clashing updates, i.e., if for every pair of updates $(l_i, b_i)$ and $(l_j, b_j)$ in $\Delta$, we have that $l_i = l_j$ only if $b_i = b_j$.    
\end{definition}

A consistent set of updates $\Delta$ is \emph{executed} (or \emph{fired}) by executing all the updates in $\Delta$ simultaneously.

\begin{definition}[Execution of Updates]
The result of executing (or firing) a consistent update set $\Delta$ in a state ${\bf S}$ is a new state ${\bf S} + \Delta$ with the same base set as $\bf S$ such that for every location $l_i = (f_i, \bar{a}_i)$ of $\bf S$:
\[f_i^{{\bf S}+\Delta}(\bar{a}_i) = 
\begin{cases}
b & \text{if } (l_i, b) \in \Delta; \\
f_i^{{\bf S}}(\bar{a}_i) & \text{if there is no } b \text{ with } (l_i,b) \in \Delta.
\end{cases}
\]
\end{definition}

If ${\bf S}_1$ and ${\bf S}_2$ are structures of the same vocabulary and with the same base set, then there is a unique consistent set $\Delta$ of non-trivial updates of ${\bf S}_1$ such that ${\bf S}_2 = {\bf S}_1 + \Delta$. We use ${\bf S}_2 - {\bf S}_1$ to denote this unique consistent set of updates $\Delta$.

The following well known lemma is a consequence of the fact that by Postulate~\ref{abstractState} the one step transformation function $\tau_A$ of an algorithm $A$ is required to preserve isomorphisms. 

\begin{lemma}\label{lemma:IsoExtendToUpdateSets}
Suppose that $\zeta$ is an isomorphism from state ${\bf S}_1$ to state ${\bf S}_2$ of an algorithm $A$. We extend $\zeta$ to locations and update sets of ${\bf S}_1$ as follows:
\begin{itemize}
\item If $l = (f, (a_1, \ldots, a_n))$ is a location in ${\bf S}_1$, then\\ $\zeta(l) = (f, (\zeta(a_1), \ldots, \zeta(a_n)))$.
\item If $\Delta$ is an update set for ${\bf S}_1$, then $\zeta(\Delta) = \{(\zeta(l), \zeta(v)) \mid (l, v) \in \Delta\}$.
\end{itemize}
Then $\zeta(\tau_A({\bf S}_1) - {\bf S}_1) = \tau_A({\bf S}_2) - {\bf S}_2$.
\end{lemma}

As usual, we assume that a function symbol can be marked (by the vocabulary) as \emph{relational} and that every vocabulary includes the binary function symbol ``$=$'' for equality, nullary function symbols $\texttt{true}$, $\texttt{false}$ and $\texttt{undef}$, the unary function symbol \texttt{Boole}, and the symbols ``$\neg$'', ``$\vee$'', ``$\wedge$'' and ``$\rightarrow$'' corresponding to the usual Boolean operations. With the exception of $\texttt{undef}$, all these \emph{logic symbols} are relational. Furthermore, all of them are marked as static.

We identify a  nullary function with its value in the base set of the state. Thus, if ${\bf S}$ is a first-order structure, $\texttt{true}^{\bf S}$, $\texttt{false}^{\bf S}$ and $\texttt{undef}^{\bf S}$ are particular elements of $S$. 
We require that $\texttt{true}^{\bf S}$ be a distinct element than $\texttt{false}^{\bf S}$ and $\texttt{undef}^{\bf S}$. 

An $r$-ary relation $R^{\bf S} \subseteq S^r$ is \emph{represented} by a relational function $f_R^{\bf S}$ from $S^r$ to $\{\texttt{true}^{\bf S},\texttt{false}^{\bf S}\}$ such that, for every $\bar{a} \in S^r$, it holds that $\bar{a} \in R^{\bf S}$ iff $f_R^{\bf S}(\bar{a}) = \texttt{true}^{\bf S}$. If a relation $R^{\bf S}$ is unary, it can be viewed as a (sub)set of elements from the domain of ${\bf S}$. Thus, in any given structure $\bf S$, $\texttt{Boole}$ is interpreted as the set $\{\texttt{true}^{\bf S},\texttt{false}^{\bf S}\}$. The Boolean operations behave in the usual way on (the interpretation of) $\texttt{Boole}$ and produce $\texttt{false}$ if at least one of its arguments is not Boolean. 

We assume that all the functions are total and represent partial function as total by using $\texttt{undef}$. Thus the \emph{domain} of a non-relational $r$-ary function $f^{\bf S}$ of a structure $\bf S$ is the set $\{\bar{a} \in S^r \mid f^{\bf S}(\bar{a}) \neq \texttt{undef}^{\bf S}\}$. The \emph{range} of $f^{\bf S}$ is the set $\{f^{\bf S}(\bar{a}) \mid \bar{a} \in S^r \text{ and } f^{\bf S}(\bar{a}) \neq \texttt{undef}^{\bf S}\}$.

When it is clear from the context, we sometimes use function symbols to denote their interpretations, i.e., we omit the superscripts. We also favour infix notation for certain functions such as the (binary) Boolean operations and equality.

\section{Background of Computation}
\label{sec:backgroundOfComputation}

The background of computation that we need for our work is in essence the same as in Blass and Gurevich's formalization of parallel algorithms \cite{[BG03],[BG08]}. However, our formalization is different in two key aspects. First we consider states as meta-finite structures. This makes the intrinsic finiteness of the states of the algorithms explicit, allowing for a natural representation of the finite parts of the states. In particular, the technicality of assuming that every state includes a finite set of processes (i.e., proclets in \cite{[BG03],[BG08]}) is not longer needed. Secondly, we make an explicit distinction between pure states of the algorithms and states of the \emph{computation} of the algorithms. This latter notion of states includes, apart from the state of the algorithm, the standard background needed for its computation.

\subsection{Meta-Finiteness of States}

States of algorithms are intrinsically finite. However, as pointed out in \cite{[BG00]} there is more to a computation than what is just available in the states; this gives rise to the augmentation of states with background structures and to infinite representations of states. In fact, this is the norm in the ASM literature (see \cite{[BS03]}  among others). Take as an example input to an algorithm composed by a graph and a weight function from the edges of the graph to the natural numbers. There are many ways to represent a state of an algorithm that includes this input as a finite structure. One could for instance replace every edge $(u, v)$ of weight $w$ by $w$ distinct nodes, each connected to $u$ and $v$ but to no other nodes. While the resulting finite structure contains all the information about the original weighted graph, it is very impractical to perform arithmetic operations involving the encoded weights or to perform verification proofs of the relevant algorithm. 

On the other hand, if we represent states as arbitrary (possibly infinite) first-order structures, then the intrinsic finiteness of the states of an algorithm can only be captured implicitly, for instance by assuming a background with a variable-free term that evaluates to a finite set such as the term \texttt{Proclet} in the background postulate in \cite{[BG03]}. A more faithful representation would be to have an auxiliary infinite structure, which in our example could be the set of natural numbers with the usual arithmetic operations. This prompted us to consider states as meta-finite structures as defined in \cite{[GG98]}.

\begin{definition} 

\ A \emph{meta-finite structure} $\bf I$ is a triple $({\bf I}_1, {\bf I}_2, F)$ where

\begin{enumerate}[label=\roman{enumi}.]

\item ${\bf I}_1$ is a finite first-order structure -- the \emph{primary part} of $\bf I$;
\item ${\bf I}_2$ is a possibly infinite first-order structure -- the \emph{secondary part} of $\bf I$;
\item $F$ is a finite set of functions $f_i : (I_1)^k \longrightarrow I_2$ -- the \emph{bridge functions} of $\bf I$.

\end{enumerate}

The \emph{vocabulary} of ${\bf I}$ is the triple $\Sigma_{\bf I} = (\Sigma_1, \Sigma_2, \Sigma_F)$ where $\Sigma_1$, $\Sigma_2$ and $\Sigma_F$ are the (pairwise disjoint) sets of function symbols in ${\bf I}_1$, ${\bf I}_2$ and $F$, respectively. The \emph{base set} $I$ of ${\bf I}$ is $I_1 \cup I_2$.

\end{definition}

\begin{example}\label{Ex:StateOfAlgorithm}
Weighted graphs could be represented as meta-finite structures of vocabulary $\Sigma = (\emptyset, \Sigma_2, \Sigma_F)$, where $\Sigma_2$ includes all the (background) function symbols described in Section~\ref{Sec:Prelim} plus function symbols for the standard arithmetic operations, and  $\Sigma_F$ has two binary function symbols $f_E$ and $f_w$ with $f_E$ marked as relational. For instance, let $G$ be a digraph with vertex set $V = \{a, b, c\}$ and edge set $E = \{(a, b), (b, c), (c, a)\}$, and $w: E \longrightarrow \mathbb{N}$ be the weight function $\{(a, b) \mapsto 3, (b, c) \mapsto 5, (c, a)\mapsto 7\}$. Then $G$ can be represented by a meta-finite state ${\bf I} = ({\bf I}_1, {\bf I}_2, F)$ of vocabulary $\Sigma$ in which $I_1$ (the base set of ${\bf I}_1$) is $V$, $I_2$ (the base set of ${\bf I}_2$) is the set $\mathbb{N}$ of natural numbers, $f_E^{\bf I}(x, y) = \texttt{true}^{\bf I}$ iff $(x,y) \in E$, and $f_w^{\bf I}(x, y) = z$ if $w(x, y) = z$ or  $f_w^{\bf I}(x, y) = \texttt{undef}^{\bf I}$ otherwise
(recall that we assume that all functions are total).
\end{example}

\subsection{States of a Computation}  

During a computation, algorithms frequently need to deal with constructions that produce new elements (e.g., tuples or multisets) from old ones (the components of a tuple or the elements of a multiset). Such constructions can be iterated, producing tuples of multisets, multisets of tuples, and so forth. A general approach to this kind of constructions was developed in \cite{[BG00]} under the name of ``Background Classes''. These classes formalize the idea of things that can be built on top of a set without introducing any additional structure to the set itself. Formally, background classes are determined by background vocabularies that consist of constructor symbols and function symbols. Different to function symbols of fixed arity, constructor symbols can also be of bounded or even unfixed arity. 

\begin{definition}
\label{def:background}
A \emph{background class} $\mathcal{K}$ of vocabulary $\Sigma_{\cal K}$ associates with any set $U$ a \emph{background structure} $\mathcal{K}(U)$ constituted by  \begin{itemize}
\item the base set $\mathrm{Base}(\mathcal{K}(U)) = D$, where $D$ is the smallest set with $U \subseteq D$ that satisfies the following properties for each constructor symbol $\llcorner\lrcorner \in \Sigma_{\cal K}$:
\begin{itemize}
\item if $\llcorner\lrcorner \in \Sigma_{\mathcal{K}}$ has unfixed arity, then $\llcorner a_1, \ldots, a_m \lrcorner \in D$ for all $m \in \mathbb{N}$ and $a_1, \ldots, a_m \in D$.
\item if $\llcorner\lrcorner \in \Sigma_{\mathcal{K}}$ has bounded arity $n$, then $\llcorner a_1, \ldots, a_m\lrcorner \in D$ for all $m \leq n$ and $a_1, \ldots, a_m \in D$.
\item if $\llcorner \lrcorner \in \Sigma_{\mathcal{K}}$ has fixed arity $n$, then $\llcorner a_1, \ldots, a_n\lrcorner \in D$ for all $a_1, \ldots, a_n \in D$.
\end{itemize}
\item an interpretation of function symbols in $\Sigma_{\cal K}$ over $\mathrm{Base}(\mathcal{K}(U))$.
\end{itemize}
\end{definition}

Summing up, a \emph{state of an algorithm} $A$ can be thought of as a simple meta-finite structure $\bf I$. But, the initial \emph{state of computation} of $A$ is actually richer, as it should also include the background structure ${\cal K}(I)$ corresponding to the background class ${\cal K}$ of the algorithm $A$. 
 
\begin{definition}[State of Computation]\label{Def:CompState}
Let $A$ be an algorithm of \emph{vocabulary} $\Sigma_{\bf I} = (\Sigma_1, \Sigma_2, \Sigma_F)$ with background class $\cal K$ of vocabulary $\Sigma_{\cal K}$, where $\Sigma_{\cal K}$, $\Sigma_1$, $\Sigma_2$ and $\Sigma_F$ are pairwise disjoint. Let $\texttt{atomic}$ be a unary relation symbol which does not belong to $\Sigma_{\bf I} \cup \Sigma_{\mathcal{K}}$. Let ${\bf I} = ({\bf I}_1, {\bf I}_2, F)$ be a state of $A$, i.e., a meta-finite structure of vocabulary $\Sigma_{\bf I}$ that represents a valid state of $A$. Let ${\bf S} = ({\bf I}_1, {\bf B}, F)$ be the meta-finite structure of vocabulary $\Sigma = (\Sigma_1, \Sigma_2 \cup \Sigma_{\cal K} \cup \{\texttt{atomic}\}, \Sigma_F)$, where ${\bf B}$ is the first-order structure of vocabulary $\Sigma_2 \cup \{f_i \mid f_i \text{ is a function symbol in } \Sigma_{\cal K}\} \cup \{\texttt{atomic}\}$ which satisfies the following conditions: 
\begin{enumerate}[label=\roman{enumi}.]
\item The base set $B$ of ${\bf B}$ is $\mathrm{Base}(\mathcal{K}(I))$ (recall that $I$ denotes the base set of ${\bf I}$).
\item $\texttt{atomic}^{\bf B}(x) = \texttt{true}^{\bf B}$ iff $x \in I$ (recall that we assume in Section~\ref{Sec:Prelim} that every state includes the constants $\texttt{true}$, $\texttt{false}$ and $\texttt{undef}$). 
\item For every function symbol $f_i \in \Sigma_{\cal K}$ it holds that $f_i^{\bf B} = f_i^{{\cal K}(I)}$.
\item For every $f_i \in \Sigma_2$ of arity $r_i$ and $\bar{a} \in B^{r_i}$, it holds that
\[
f_i^{\bf B}(\bar{a}) =
\begin{cases}
f_i^{\bf I}(\bar{a}) & \text{if } \bar{a} \in ({I_2})^{r_i} \\
\texttt{false}^{\bf B} & \text{if } \bar{a} \not\in ({I_2})^{r_i} \text{ and } f \text{ is marked as relational} \\
\texttt{undef}^{\bf B} & \text{otherwise}\\
\end{cases}
\]
\end{enumerate}
We say that ${\bf S}_{\bf I}$ is a \emph{state of computation} of $A$ that corresponds to the state $\bf I$ if it is isomorphic to ${\bf S}$ by an isomorphism $\zeta$ such that $\zeta(x) = x$ for all $x \in I$. 
\end{definition}

\begin{example}
Let $\Sigma_{\cal K}$ be the background vocabulary formed by a constructor symbol for pairing plus unary function symbols \texttt{first} and \texttt{second} interpreted as the functions mapping pairs to their first and second element, respectively. If the argument is not a pair, then both functions map it to $\texttt{undef}$. A state of computation corresponding to the state ${\bf I}$ of the algorithm described in Example~\ref{Ex:StateOfAlgorithm} is represented by the meta-finite structure ${\bf S}_{\bf I} = ({\bf I}_1, {\bf B}, F)$, where the base set $B$ of ${\bf B}$ is $\{a,b,c\} \cup \mathbb{N} \cup \{(a_i, a_j) \mid a_i, a_j \in B\}$, $\{x \mid \texttt{atomic}^{\bf B}(x)\} = \{a,b,c\} \cup \mathbb{N}$, $\texttt{first}^{\bf B} = \texttt{first}^{{\cal K}(\{a, b, c \} \cup \mathbb{N})}$ and $\texttt{second}^{\bf B} = \texttt{second}^{{\cal K}(\{a, b, c \} \cup \mathbb{N})}$. The remaining functions in ${\bf B}$, i.e., the functions corresponding to the standard arithmetic operations and the logic symbols, coincide with their corresponding functions in ${\bf I}_2$ when all their arguments belong to $\mathbb{N}$ and take otherwise the value $\texttt{false}^{\bf B}$ or $\texttt{undef}^{\bf B}$, depending on whether the function is relational or not.    
\end{example}

Given a state of computation ${\bf S} = ({\bf I}, {\bf B}, F)$ of an algorithm $A$, we call the functions in ${\bf B}$ the \emph{background functions} of ${\bf S}$ and the functions in ${\bf I}$ or $F$ the \emph{foreground functions} of ${\bf S}$. Consequently, we call ${\bf B}$ the \emph{background of computation} of $A$. As expected, all background functions are assumed to be static with the only exception of the unary function $\texttt{reserve}$ defined next. 

\subsection{The Reserve}\label{Sec:theReserve}

It is justified to assume that the base set of a state does not change during a computation. To realize this, it is convenient to include in every state an infinite supply of \emph{reserve} elements that can be imported by an algorithm, when new elements are needed. The reserve is a ``naked set'', i.e., an entirely unstructured part of the base set. 

\begin{definition}\label{Def:theReserve}
Let $A$ be an algorithm of vocabulary $\Sigma_{\bf I} = (\Sigma_1, \Sigma_2, \Sigma_F)$. We assume that $\Sigma_2$ includes a unary function symbol $\texttt{reserve}$ which is marked as relational. We also assume that the interpretation of $\texttt{reserve}$ in every state ${\bf I} = ({\bf I}_1, {\bf I}_2, F)$ of $A$ satisfies the following conditions:
\begin{enumerate}[label=\roman{enumi}.]
\item $\texttt{reserve}^{{\bf I}_2}(x) = \texttt{true}^{{\bf I}_2}$ iff $x \not\in I_1$ and, for every $f_i \in \Sigma_2$ of arity $r_i$ and tuple $\bar{a}_i \in (I_2)^{r_i}$ which includes $x$, it holds that $f_i^{{\bf I}_2}(\bar{a}_i)$ evaluates to $\texttt{false}^{{\bf I}_2}$ if $f_i$ is marked as relational or to $\texttt{undef}^{{\bf I}_2}$ otherwise. 
\item $R = \{a_i \in  I_2 \mid \texttt{reserve}^{{\bf I}_2}(a_i) = \texttt{true}^{{\bf I}_2}\}$ has countably many elements.
\end{enumerate}
We call $R$, the \emph{reserve} of the state ${\bf I}$ of $A$. We also call $R$, the reserve of its corresponding state of computation as per Definition~\ref{Def:CompState}.    
\end{definition}

Reserve elements can be imported when, for example, an algorithm needs to add a new vertex to a graph. To import a new element from the reserve basically involves to take it from the reserve and to add it to the primary (finite) part of the state.   

\begin{definition} 
Let $A$, ${\bf I}$ and $R$ be as in Definition~\ref{Def:theReserve}. An element $a_i \in R$ is \emph{imported} from the reserve by simply adding it to the base set $I_1$ of the primary part of ${\bf I}$ and updating the value of $\texttt{reserve}^{{\bf I}_2}(a_i)$ to $\texttt{false}^{{\bf I}_2}$.  
\end{definition}

\begin{remark}
Note that base set of the primary (finite) part of the state of an algorithm can only grow monotonically during a computation. That is, we can import new elements from the reserve into the base set of the primary part of a state, but we cannot discard elements from it. The base set of the secondary part (i.e., the background of computation) as well as the base set of the whole state, remains the same. 
\end{remark}

\subsection{Minimal Background Requirement}
\label{sec:minBackground}

In order to simulate a given parallel algorithm, we need a minimal background, which essentially should contain ordered pairs and multisets. 
Multiplicities arise naturally in parallel algorithms, mainly from the repetition of tasks in the different parallel branches. For instance, if several branches request to increment the same counter in parallel, then during the synchronization phase the algorithm should take into account the multiplicity of requests concerning the same counter in order to increment it by the correct amount.  Multisets are an appropriate tool to collect these multiplicities. In turn, multiset operators are useful for synchronization. In the example of the counter, the synchronization phase would consist in applying a static function ``sum'' to a multiset to obtain the sum of its members including multiplicities.

Formally, a \emph{multiset} $M$ can be seen as a function from the underlying set $M_0$ of elements in $M$ (the domain of $M$) to the positive integers, such that $M(x)$ is the multiplicity of $x$ as an element of $M$. We use $\text{Mult}(x, M)$ to denote the multiplicity $M(x)$ of an element $x$ in a multiset $M$. If $x \not\in M$ then $\text{Mult}(x, M) = 0$. We use double braces $\multiset{\ldots}$ as notation for multisets. We define \emph{binary multiset union} $M_1 \uplus M_2$ of two multisets $M_1$ and $M_2$ by $\text{Mult}(x, M_1 \uplus M_2) = \text{Mult}(x, M_1) + \text{Mult}(x, M_2)$. Consequently, we define the \emph{generalized multiset union} of a multiset of multisets ${\cal M}$, i.e. a multiset ${\cal M}$ whose underlying domain ${\cal M}_0$ is a set which contains only multisets, by $\text{Mult}(x, \biguplus {\cal M}) = \sum_{M_i \in {\cal M}_0} \text{Mult}(x, M_i) \cdot \text{Mult}(M_i, {\cal M})$.

Our next postulate defines the minimum background. This minimum background is almost the same as the one used in Blass and Gurevich's work \cite{[BG03],[BG08]}, except that we no longer need a variable-free term \texttt{Proclet} naming a finite set of computation branches. In our work, this finiteness is instead captured by the use of meta-finite structures to represent the states of an algorithm.

\begin{postulate}[Background Postulate]
\label{backgroundPostulate}\rm

\ Let $A$ be an algorithm of \emph{vocabulary} $\Sigma = (\Sigma_1, \Sigma_2, \Sigma_F)$ with background class $\cal K$.
The vocabulary  $\Sigma_{\cal K}$  of ${\cal K}$ includes (at least) a binary \emph{tuple constructor} and a \emph{multiset constructor} of unbounded arity; and the vocabulary $\Sigma_{\bf B}$ of the background of the computation states of $A$ includes (at least) the following \emph{obligatory} function symbols:
\begin{itemize}
\item Nullary function (constants) symbols $\texttt{true}$, $\texttt{false}$, $\texttt{undef}$ and $\oslash$.
\item Unary function symbols $\texttt{reserve}$, $\texttt{atomic}$, $\texttt{first}$, $\texttt{second}$, $\texttt{Boole}$, $\neg$, $\multiset{\cdot}$, $\biguplus$ and  $\texttt{AsSet}$.
\item Binary function symbols $=$, $\wedge$, $\vee$, $\rightarrow$, $\leftrightarrow$, $\uplus$ and $(\,,)$.
\end{itemize}
All function symbols in $\Sigma_{\bf B}$, with the sole exception of $\texttt{reserve}$, are static.
  Let ${\bf I} = ({\bf I}_1, {\bf I}_2, F)$ be a state of $A$. The interpretation of the obligatory function symbols in the vocabulary $\Sigma_{\bf B}$ of the secondary part (background) of every state of computation ${\bf S}_{\bf I} = ({\bf I}_1, {\bf B}, F)$ corresponding to ${\bf I}$ (see Definition~\ref{Def:CompState}), satisfies the following conditions:
\begin{itemize}
\item $\texttt{atomic}^{\bf B}(x) = \texttt{true}^{\bf B}$ iff $x \in I$.
\item $\texttt{reserve}^{\bf B}(x)= \texttt{true}^{\bf B}$ iff $x \in \texttt{reserve}^{{\bf I}}(x)$.
\item $\texttt{true}^{\bf B} = \texttt{true}^{\bf I}$, $\texttt{false}^{\bf B} = \texttt{false}^{\bf I}$ and $\texttt{undef}^{\bf B} = \texttt{undef}^{\bf I}$. Further, $\texttt{true}^{\bf B}$ is a distinct element than $\texttt{false}^{\bf B}$ and $\texttt{undef}^{\bf B}$.
\item $\texttt{Boole}^{\bf B}(x) = \texttt{true}^{\bf B}$ iff $x = \texttt{true}^{\bf B}$ or $x = \texttt{false}^{\bf B}$.
\item $=^{\bf B}$ is the identity relation in $B$.
\item $\neg^{\bf B}$, $\wedge^{\bf B}$, $\vee^{\bf B}$, $\rightarrow^{\bf B}$ and $\leftrightarrow^{\bf B}$ behave in the usual way in $\{\texttt{true}^{\bf B}, \texttt{false}^{\bf B}\}$ and produce $\texttt{false}^{\bf B}$ if at least one of its arguments is not Boolean.
\item $(x,y)^{\bf B}$ evaluates to the ordered pair with fist element $x$ and second $y$.
\item $\texttt{first}^{\bf B}(x)$ evaluates to the first element of $x$ and $\texttt{second}^{\bf B}(x)$ evaluates to the second element of $x$ if $x$ is an ordered pair, or to $\texttt{undef}^{\bf B}$ otherwise.
\item $\oslash^{\bf B}$ is the empty multiset.
\item $\multiset{x}^{\bf B}$ evaluates to the singleton multiset whose only element is $x$.
\item $\texttt{AsSet}^{\bf B}(x)$ evaluates to the multiset obtained from $x$ by setting to $1$ the multiplicity of every element in $x$. If $x$ is not a multiset, then $\texttt{AsSet}^{\bf B}(x)$ evaluates to $\texttt{undef}^{\bf B}$.
\item $x \uplus^{\bf B} y$ evaluates to the binary multiset union of $x$ and $y$. If $x$ or $y$ is not a multiset, then $x \uplus^{\bf B} y$ evaluates to $\texttt{undef}^{\bf B}$.
\item $\biguplus^{\bf B} x$ evaluates to the generalized multiset union of the multisets in $x$ if $x$ is a multiset whose elements are all multisets. Otherwise it evaluates to $\texttt{undef}^{\bf B}$.
\end{itemize}
\end{postulate}

\section{Bounded Exploration for Parallel Algorithms}
\label{sec:boundedExploration}

We now introduce our fourth and last postulate, namely the Bounded Exploration postulate for parallel algorithms. This is the key postulate in our work. It replaces the Proclet Algorithm, Bounded Sequentiality and Update postulates in the parallel ASM thesis of Blass and Gurevich \cite{[BG03],[BG08]}. 

First we define the formal syntax and semantics of the terms that we require to state our postulate. These terms coincide with the terms used by the  ASM model for parallel computation considered in this paper. They are built up from variables, functions and a multiset comprehension expression which is similar to the multiset comprehension expression used in the definition of the ASM model in the parallel thesis of Blass and Gurevich.

We follow the standard approach in \emph{meta-finite} model theory \cite{[GG98]} assuming that that variables range over the domain of the primary part only and considering two different types of basic terms: \emph{point terms} which define functions over the primary part of a meta-finite state, and \emph{bridge terms} which define functions that take arguments in the primary part of a meta-finite state and values in the secondary part. The actual set of terms and its semantics are included in the following definition. If the outermost function symbol of a term $\varphi$ is relational, we call it a \emph{Boolean-valued term}.

\begin{definition}
\label{Def:terms}
Let ${\bf S} = ({\bf I}, {\bf B}, F)$ be a (meta-finite) \emph{state of computation} of an algorithm $A$ with background class ${\cal K}$, which includes a multiset constructor (see Definition~\ref{Def:CompState}). Let $\Sigma = (\Sigma_{\bf I}, \Sigma_{\bf B}, \Sigma_F)$ and $\Sigma_{\cal K}$ be the vocabularies of ${\bf S}$ and ${\cal K}$, respectively. Let $V = \{x_0, x_1,  \ldots\}$ be a countable set of variables. The set of terms ${\cal T}_{\Sigma, V}$ over $\Sigma$ and $V$ is defined inductively as follows:
\begin{itemize}
\item The set of \emph{point terms} is the closure of the set $V$ of variables under the application of function symbols in $\Sigma_{\bf I}$. Every point term belongs to ${\cal T}_{\Sigma, V}$.
\item If $t_1, \ldots t_r$ are point terms in ${\cal T}_{\Sigma, V}$ and $f$ is an $r$-ary function symbol in $\Sigma_F$, then $f(t_1, \ldots, t_r)$ is a \emph{bridge term} in ${\cal T}_{\Sigma, V}$.
\item If $F_1, \ldots, F_r$ are bridge terms in ${\cal T}_{\Sigma, V}$ and $f$ is an $r$-ary function symbol in $\Sigma_{\bf B}$, then $f(F_1, \ldots, F_r)$ is a \emph{bridge term} in ${\cal T}_{\Sigma, V}$.
\item Let $\bar{x}$ and $\bar{y}$ be tuples of variables. If $t(\bar{x},\bar{y})$ is a term in ${\cal T}_{\Sigma, V}$ and $\varphi(\bar{x},\bar{y})$ is a Boolean valued term also in ${\cal T}_{\Sigma, V}$, then $\multiset{t(\bar{x}, \bar{y}) \mid \varphi(\bar{x}, \bar{y})}_{\bar{x}}$ is a \emph{multiset comprehension term} in ${\cal T}_{\Sigma, V}$ with free variables $\bar{y}$.
\end{itemize}
Let $\mu$ be a \emph{variable assignment} for $V$ over the primary part ${\bf I}$ of ${\bf S}$, i.e., a function which assigns to each variable $x_i$ in $V$ a value $\mu(x_i) \in I$. Let $t$ be a term in $\mathcal{T}_{\Sigma,V}$. The value $\text{val}_{{\bf S},\mu}(t)$ of $t$ in ${\bf S}$ under $\mu$ is inductively defined as follows:
\begin{itemize}
\item If $t$ is a nullary function symbol in $\Sigma$, then $\text{val}_{{\bf S},\mu}(t) = t^{\bf S}$.
\item If $t$ is a variable in $V$, then $\text{val}_{{\bf S},\mu}(t) = \mu(t)$.
\item If $t$ is a function symbol in $\Sigma$ of arity $r \geq 1$ and $t_1, \ldots, t_r$ are terms in $\mathcal{T}_{\Sigma,X}$, then $\text{val}_{{\bf S},\mu}(t(t_1,\ldots,t_r)) = t^{\bf S}(\text{val}_{{\bf S},\mu}(t_1), \ldots, \text{val}_{{\bf S},\mu}(t_r))$.
\item Let $\bar{x} = (x_1, \ldots, x_n)$ and $\bar{y} = (y_1, \ldots, y_m)$ for some $n$ and $m \geq 0$. If $t(\bar{y})$ is a multiset comprehension term of the form $\multiset{s(\bar{x}, \bar{y}) \mid \varphi(\bar{x}, \bar{y})}_{\bar{y}}$, then
$\text{val}_{{\bf S},\mu[\bar{y} \mapsto \bar{b}]}(t(\bar{x}, \bar{y}))$ is the multiset\\[0.2cm]
$\{\!\!\{val_{{\bf S},\mu[\bar{x} \mapsto \bar{a}, \bar{y} \mapsto \bar{b}]}(s(\bar{x},\bar{y})) \mid val_{{\bf S},\mu[\bar{x} \mapsto \bar{a}, \bar{y} \mapsto \bar{b}]}(\varphi(\bar{x},\bar{y})) = \texttt{true}^{\bf S} \text{ and } \bar{a} \in I^n\}\!\!\}$
\end{itemize}
\end{definition}

\begin{newtext}
As in the examples of this article, function symbols may be applied to arbitrary terms, and comprehension terms may occur as arguments of background functions; symbols of $\Sigma_{\bf I}$ and $\Sigma_F$ take the value $\texttt{false}$ (if relational) or $\texttt{undef}$ on arguments outside the primary part, and the logic symbols are interpreted on the whole base set as stated in the Background postulate, so that $=$ is the identity.
\end{newtext}

If $t$ is a ground term, we simply use $\text{val}_{\bf S}(t)$ to denote its value in ${\bf S}$.  Given a Boolean valued term $\varphi(x_1, \ldots, x_r)$ with free variables among $\{x_1, \ldots, x_r\}$, we frequently use ${\bf S} \models \varphi(x_1, \ldots, x_r)[a_1, \ldots, a_r]$ to denote that
\[\text{val}_{{\bf S},\mu[x_1 \mapsto a_1, \ldots, x_r \mapsto a_r]}(\varphi(x_1, \ldots, x_r)) = \texttt{true}^{\bf S}.\]
We use $\exists x_1 \ldots x_r ( \varphi(x_1,$ $\ldots, x_r) )$ and $\forall x_1 \ldots x_r ( \varphi(x_1, \ldots, x_r) )$ to denote that\\[0.2cm]
$\multiset{(x_1, \ldots , x_r) \mid \varphi(x_1, \ldots, x_r)} \neq \oslash$ and\\[0.2cm] 
$\multiset{(x_1, \ldots , x_r) \mid \neg(\varphi(x_1, \ldots, x_r))} = \oslash$, respectively.\\[0.2cm]
We refer to the class of multiset comprehension terms of the form \[\multiset{t(x_1, \ldots, x_r) \mid \varphi(x_1, \ldots, x_r)}\] which have \emph{no} free-variables and where $t$ is an ordered pair that represents (using some fixed encoding) a tuple $(t_0, \ldots, t_n)$ of terms with $\textit{free}(t_0) \cup \cdots \cup \textit{free}(t_n) = \{x_1, \ldots, x_r\}$, as \emph{witness terms}. Consequently, we denote them as $\multiset{(t_1, \ldots, t_n) \mid \varphi(x_1, \ldots, x_r)}$. 

\begin{definition}
Two states of computation ${\bf S}_1$ and ${\bf S}_2$ of a same vocabulary $\Sigma$ \emph{coincide} over a set $W$ of witness terms if for every $\alpha_i \in W$,  we have that $\text{val}_{{\bf S}_1}(\alpha_i) = \text{val}_{{\bf S}_2}(\alpha_i)$.
\end{definition}

\begin{newtext}
Let $W=\{\alpha_1,\ldots,\alpha_m\}$ be a finite set of witness terms with
\[
\alpha_i=\multiset{(t_{i0},\ldots,t_{in_i})\mid\varphi_i}_{x_1,\ldots,x_{k_i}}\qquad(1\le i\le m),
\]
and let ${\bf S}$ be a state of computation whose primary part has base set $I$ (Definition~\ref{Def:terms}). An assignment $\bar c\in I^{k_i}$ is \emph{witnessing} for $\alpha_i$ in ${\bf S}$ if ${\bf S}\models\varphi_i[\bar c]$; its \emph{value tuple} is $(b_0,\ldots,b_{n_i})$ with $b_j=\val_{{\bf S},\mu[x_1\mapsto c_1,\ldots,x_{k_i}\mapsto c_{k_i}]}(t_{ij})$. Thus $\val_{\bf S}(\alpha_i)$ is the multiset of the value tuples of the witnessing assignments of $\alpha_i$.

\begin{definition}[Critical Values]\label{def:criticalValues}
The set $V_{{\bf S},W}$ of \emph{critical values} of ${\bf S}$ with respect to $W$ is the set of the components of the tuples in the multisets $\val_{\bf S}(\alpha_1),\ldots,\val_{\bf S}(\alpha_m)$. A \emph{critical tuple} is a tuple of critical values.
\end{definition}

\begin{definition}[Critical Structure]\label{def:criticalStructure}
Let $\Sigma_W=\{R_{\alpha_1},\ldots,R_{\alpha_m}\}$, where each relation symbol $R_{\alpha_i}$ has arity $n_i+2$. The \emph{critical structure} ${\bf S}|_W$ of ${\bf S}$ with respect to $W$ is the finite relational structure of vocabulary $\Sigma_W$ with
\[
\begin{aligned}
R^{{\bf S}|_W}_{\alpha_i}=\{(b_0,\ldots,b_{n_i},\str_i(\bar c))\mid\ &\bar c\text{ is witnessing for }\alpha_i\text{ in }{\bf S}\\
&\text{and has value tuple }(b_0,\ldots,b_{n_i})\}.
\end{aligned}
\]
Here the \emph{string} $\str_i(\bar c)=(i,b_0,\ldots,b_{n_i},c_1,\ldots,c_{k_i})$ is a formal tuple, a new object that is not an element of ${\bf S}$; it determines $i$ and $\bar c$. The domain of ${\bf S}|_W$ consists of the critical values and the strings.
\end{definition}

In the case of critical structures, we do \emph{not} assume the existence of an equality relation, since we do not want to include anything outside of what is prescribed by the witness set. The strings have the sole purpose of encoding the multiplicities of the elements in the multisets $\val_{\bf S}(\alpha_i)$, again without including anything outside of what is prescribed by the witness set.

\begin{definition}[Uniformly Definable Updates]\label{def:udu}
Let $A$ be an algorithm and $W$ a finite set of witness terms. $W$ \emph{covers the updates} of $A$ if for every state ${\bf I}$ of $A$, every corresponding state of computation ${\bf S}$ and every update $((f,(a_1,\ldots,a_r)),a_0)\in\tau_A({\bf I})-{\bf I}$, the tuple $(a_0,\ldots,a_r)$ is a critical tuple of ${\bf S}$. $A$ has \emph{uniformly definable updates} with respect to $W$ if $W$ covers the updates of $A$ and for every dynamic function symbol $f$ of arity $r$ there is an equality-free first-order formula $\psi_f(x_0,\ldots,x_r)$ of vocabulary $\Sigma_W$ such that for every state ${\bf I}$ of $A$, every corresponding state of computation ${\bf S}$ and every critical tuple $\bar a=(a_0,\ldots,a_r)$ of ${\bf S}$:
\[
((f,(a_1,\ldots,a_r)),a_0)\in\tau_A({\bf I})-{\bf I}\iff{\bf S}|_W\models\psi_f[\bar a].
\]
\end{definition}
\end{newtext}

\begin{postulate} [Bounded Exploration Postulate]
\label{boundedexpl}\rm
Let $A$ be a parallel algorithm. Then there is a finite set $W$ of witness terms, called \emph{bounded exploration witness} of $A$, such that for every pair of states ${\bf I}_1$ and ${\bf I}_2$ of $A$, it holds that $\tau_A({\bf I}_1) - {\bf I}_1 = \tau_A({\bf I}_2) - {\bf I}_2$ whenever there are two computation states ${\bf S}_1$ and ${\bf S}_2$ corresponding to ${\bf I}_1$ and ${\bf I}_2$, respectively, that coincide over $W$. \new{Moreover, $A$ has uniformly definable updates with respect to $W$ (Definition~\ref{def:udu}).}
\end{postulate}

Same as in the bounded exploration postulate for sequential algorithms, the intuition is that the algorithm $A$ examines only the part of the state that is given by means of terms in $W$. The central difference resides in the fact that in the case of parallel algorithms, the terms in $W$ are multiset comprehension terms instead of ground terms. In practice, this means that it is also determined by the state, not just by the algorithm, which part of the state is actually relevant for producing the update set.

\begin{newtext}
The requirement of uniformly definable updates, which was missing in the first version of this article, makes explicit that the update set is determined by the explored part of the state in a uniform way by fixed formulae over the critical structure, which do not depend on the state. In the sequential case no such requirement is needed, because there the update set is determined by the equalities between the values of finitely many ground terms, which yields a finite case distinction. For multiset comprehension terms this is not the case, and the first condition of Postulate~\ref{boundedexpl} alone does not suffice for the characterization theorem (Remark~\ref{rem:AX}). By the following lemma, the first condition is a consequence of the second; we keep it because it expresses the intuition behind the postulate.

\begin{lemma}\label{lem:udu-coincide}
If $A$ has uniformly definable updates with respect to $W$, then $\tau_A({\bf I}_1) - {\bf I}_1 = \tau_A({\bf I}_2) - {\bf I}_2$ for all states ${\bf I}_1$ and ${\bf I}_2$ of $A$ with corresponding states of computation ${\bf S}_1$ and ${\bf S}_2$ that coincide over $W$.
\end{lemma}

\begin{proof}
Since ${\bf S}_1$ and ${\bf S}_2$ coincide over $W$, we have $V_{{\bf S}_1,W}=V_{{\bf S}_2,W}=:V$. For each $i$ and each value tuple $\bar b$, the multiplicity of $\bar b$ in $\val_{{\bf S}_1}(\alpha_i)=\val_{{\bf S}_2}(\alpha_i)$ is the number of witnessing assignments of $\alpha_i$ with value tuple $\bar b$ in ${\bf S}_1$, and likewise in ${\bf S}_2$. Choosing bijections between these assignments and mapping the strings accordingly, we obtain an isomorphism from ${\bf S}_1|_W$ onto ${\bf S}_2|_W$ that is the identity on $V$ (strings are not elements of the states, Definition~\ref{def:criticalStructure}). Hence ${\bf S}_1|_W\models\psi_f[\bar a]$ iff ${\bf S}_2|_W\models\psi_f[\bar a]$ for all $\bar a\in V^{r+1}$. Since $W$ covers the updates, both update sets consist of updates whose tuples lie in $V^{r+1}$, and uniform definability gives $\tau_A({\bf I}_1) - {\bf I}_1 = \tau_A({\bf I}_2) - {\bf I}_2$.
\end{proof}

\begin{remark}\label{rem:AX}
The second condition of Postulate~\ref{boundedexpl} is not implied by the first. For $X\subseteq\mathbb N$, let $A_X$ be the algorithm whose states have an empty primary vocabulary and two unary, dynamic, relational bridge functions $P$ and $Q$, whose background class is the minimal one required by the Background postulate, and which in one step sets $P(x):=\texttt{true}$ for all elements $x$ of the primary part if the size of the primary part belongs to $X$, and $Q(x):=\texttt{true}$ for all $x$ otherwise. Then $A_X$ satisfies the Sequential Time, Abstract State and Background postulates and the first condition of Postulate~\ref{boundedexpl} with
\[
W=\{\multiset{x\mid\neg P(x)},\ \multiset{x\mid\neg Q(x)},\ \multiset{x\mid\texttt{true}},\ \multiset{\texttt{true}\mid\texttt{true}}\}.
\]
If $X$ is undecidable, however, no parallel ASM with a computable background class is behaviourally equivalent to $A_X$, since the update set of such a machine in a state with $n$ primary elements would be computable from $n$. By Theorem~\ref{thm:characterization}, $A_X$ therefore violates the second condition. See~\cite{[FS-Corr]} for details.
\end{remark}
\end{newtext}

\begin{example}
\label{Ex:WitnessGraphComplement}
Let us take the parallel algorithm in Example~\ref{Ex:GraphComplement} which calculates the complement of a given undirected graph. The (parallel) bounded exploration witness\\[0.2cm]
$\{\multiset{(\neg E(x,y), x,y) \mid x \neq y \wedge V(x) \wedge V(y)},$\\[0.2cm]
$\multiset{(\neg E(x,y), x,y) \mid x \neq y \wedge \neg(V(x) \wedge V(y))},$\\[0.2cm]
$\multiset{x\neq y \mid V(x) \wedge V(y)}, \; \multiset{x\neq y \mid \neg (V(x) \wedge V(y))}, \multiset{\texttt{true} \mid \texttt{true}}\}$ \\[0.2cm]
shows that it satisfies the (parallel) bounded exploration postulate. \new{The updates of $E$ are exactly the tuples in the value of the first witness term $\alpha_1$. Hence the algorithm has uniformly definable updates, with
\[
\psi_E(x_0,x_1,x_2)=\exists s\,R_{\alpha_1}(x_0,x_1,x_2,s)
\]
and with a contradictory formula $\psi_f$ for every other dynamic function symbol~$f$.}
\end{example}

We can now formalize the concept of synchronous, parallel algorithm.

\begin{definition}
\label{def:parallelAlgo}
A \emph{(synchronous) parallel algorithm} satisfies the Sequential Time, Abstract State, Background and (Parallel) Bounded Exploration postulates.
\end{definition}

\section{Plausibility Theorem}
\label{sec:asmdef}

In this section we formally define {\em parallel} ASMs and show that every parallel ASM defines a parallel algorithm in the sense of Definition~\ref{def:parallelAlgo}.  That is, we show that every parallel ASM satisfies the Sequential Time, Abstract State, Background and (Parallel) Bounded Exploration postulates as stated in this work. Our formal ASM model is similar to the ASM model in Blass and Gurevich's thesis \cite{[BG03],[BG08]} modulo some technical details which are due to the fact that we use states of computation which are represented as meta-finite structures instead of assuming that the states include a finite set of proclets.

\begin{definition}[ASM Rules]
\label{ASMrules}
Let ${\bf I}$ be a meta-finite structure of vocabulary $\Sigma$ and let ${\bf S} = ({\bf I}, {\bf B}, F)$ of vocabulary $\Sigma_{\bf S} = (\Sigma_{{\bf I}}, \Sigma_{{\bf B}}, \Sigma_F)$ be its corresponding state of computation. Assume that ${\bf S}$ satisfies the Background postulate, let $\mu$ be a variable assignment over the primary part ${\bf I}$ of ${\bf S}$ and let $\Delta_{\mu}(r,{\bf S})$ denote the update set produced by an ASM rule $r$ in ${\bf S}$ under $\mu$. The set $\mathcal{R}$ of ASM rules over $\Sigma$ and the interpretation as update set of every rule in $\mathcal{R}$, is defined inductively by:

\begin{itemize}
\item If $f \in \Sigma$ is an $n$-ary \emph{dynamic} function symbol and $t_0, t_1, \ldots, t_n$ are terms of vocabulary $\Sigma_{\bf S}$, then $f(t_1,\ldots,t_n) := t_0$ is an \emph{assignment rule} in $\mathcal{R}$, which produces the update set \[\{((f,(\mathrm{val}_{{\bf S}, \mu}(t_1),\ldots,\mathrm{val}_{{\bf S}, \mu}(t_n))), \mathrm{val}_{{\bf S}, \mu}(t_0))\},\] provided that the following conditions hold:  

\begin{itemize}
\item  if $f \in \Sigma_{\bf I}$, then $\mathrm{val}_{{\bf S}, \mu}(t_i) \in I$ for every $0 \leq i \leq n$;
\item  if $f \in \Sigma_{\bf B}$, then $\texttt{atomic}^{\bf B}(\mathrm{val}_{{\bf S}, \mu}(t_i)) = \texttt{true}^{\bf B}$ for every $0 \leq i \leq n$;
\item  if $f \in \Sigma_F$, then $\texttt{atomic}^{\bf B}(\mathrm{val}_{{\bf S}, \mu}(t_0)) = \texttt{true}^{\bf B}$ and $\mathrm{val}_{{\bf S}, \mu}(t_i) \in I$ for every $1 \leq i \leq n$;
\end{itemize}

Otherwise, the update set is undefined.

\item If $r_1, \ldots, r_n$ are rules in $\mathcal{R}$, then $\textbf{par}\; r_1 \ldots r_n\; \textbf{endpar}$ is a \emph{block rule} in $\mathcal{R}$, which produces the update set $\Delta_{\mu}(r_1,{\bf S}) \cup \ldots \cup \Delta_{\mu}(r_n,{\bf S})$, provided all the update sets $\Delta_{\mu}(r_i,{\bf S})$ are defined.

\item If $\varphi$ is a term of vocabulary $\Sigma_{\bf S}$ and $r$ is a rule in $\mathcal{R}$, then $\textbf{if}\; \varphi\; \textbf{then}\; r$ $\textbf{endif}$ is a \emph{conditional rule} in $\mathcal{R}$, which produces the update set $\Delta_{\mu}(r,{\bf S})$ (if defined) if $\mathrm{val}_{{\bf S},\mu}(\varphi) = \texttt{true}^{\bf B}$ and the update set $\emptyset$ if $\mathrm{val}_{{\bf S},\mu}(\varphi) \neq \texttt{true}^{\bf B}$.

\item If $\varphi$ is a term of vocabulary $\Sigma_{\bf S}$ with $\textit{free}(\varphi) \supseteq \{x_1, \ldots, x_k\}$ and $r$ is a rule in $\mathcal{R}$, then $\textbf{forall}\; x_1,\ldots,x_k \; \textbf{with} \; \varphi \; \textbf{do} \; r \; \textbf{enddo}$ is a \emph{forall rule} in $\mathcal{R}$ which produces the update set $\bigcup_{(a_1, \ldots, a_k) \in A}(\Delta_{\mu[x_1 \mapsto a_1,\ldots,x_k \mapsto a_k]}(r,{\bf S}))$ where $A = \{(a_1, \ldots, a_k) \in I^k \mid \mathrm{val}_{{\bf S},\mu[x_1 \mapsto a_1,\ldots,x_k \mapsto a_k]}(\varphi) = \texttt{true}^{\bf B}\}$, provided all update sets $\Delta_{\mu[x_1 \mapsto a_1,\ldots,x_k \mapsto a_k]}(r,{\bf S})$ are defined.

\end{itemize}

\end{definition}

The scope of $x_i$ ($1 \leq i \leq k)$ in $\textbf{forall}\; x_1,\ldots,x_k \; \textbf{with} \; \varphi \; \textbf{do} \; r \; \textbf{enddo}$ is $\varphi$ and $r$. An occurrence of a variable $x$ in a transition rule $r$ is \emph{bound} if it is in the scope of a forall rule or if it is a non free variable in a multiset comprehension term. Otherwise, $x$ is \emph{free} in $r$. A rule $r$ is \emph{closed} if it has \emph{no} free variables.

\begin{definition}[Parallel ASM]
\label{ASM}
A \emph{parallel abstract state machine} $\mathcal{M}$ of vocabulary $\Sigma = (\Sigma_1, \Sigma_2,$ $\Sigma_{F})$ and background class $\mathcal{K}$ of vocabulary $\Sigma_{\cal K}$ is formed by:

\begin{itemize}

	\item A set $\mathcal{S}_{\mathcal{M}}$ of \emph{states} of vocabulary $\Sigma$ and a set $\mathcal{I}_{\mathcal{M}} \subseteq \mathcal{S}_{\mathcal{M}}$ of \emph{initial states}, both closed under isomorphisms.

	\item A closed ASM rule $r_{\mathcal{M}}$ --the \emph{main rule} of $\cal M$-- of vocabulary $(\Sigma_1, \Sigma_2 \cup \Sigma_{\cal K} \cup \{\texttt{atomic}\}, \Sigma_F)$.

	\item A transition function $\tau_{\mathcal{M}}$ over $\mathcal{S}_{\mathcal{M}}$ such that $\tau_{\mathcal{M}}({\bf I}) = {\bf I} + \Delta(r_{\mathcal{M}},{\bf S})$ for every ${\bf I} \in \mathcal{S}_{\mathcal{M}}$ and every state of computation ${\bf S}$ that corresponds to ${\bf I}$\new{, if $\Delta(r_{\mathcal{M}},{\bf S})$ is defined and consistent, and $\tau_{\mathcal{M}}({\bf I}) = {\bf I}$ otherwise; in the latter case no update is applied and the run of $\mathcal{M}$ ends, as for the empty update set}.

\end{itemize}

\end{definition}

A \emph{run} or a \emph{computation} of  an ASM ${\cal M}$ is a finite or infinite sequence ${\bf I}_0, {\bf I}_1, \ldots$, where ${\bf I}_0$ is an initial state in $\mathcal{S}_{\cal M}$ and ${\bf I}_{i+1} = \tau_{\cal M}({\bf I}_i)$ holds for every $i \geq 1$.

Next, we define for every parallel ASM $\mathcal{M}$, a finite set $W_{\cal M}$ of witness terms which, as shown in Theorem~\ref{Th:Plausibility} below, is a bounded exploration witness for $\cal M$.


\begin{newtext}
\begin{definition}
\label{def:witnessSet}
Let $\mathcal{M}$ be a parallel ASM whose main rule $r_{\mathcal M}$ has its bound variables renamed apart. Let $\nu$ range over the occurrences of assignment rules $f_\nu(t_{\nu1},\ldots,t_{\nu r_\nu}):=t_{\nu0}$ in $r_{\mathcal M}$, let $\bar x_\nu$ be the variables of the \textbf{forall} rules enclosing $\nu$, let $\gamma_\nu$ be the conjunction of the guards of the conditional and \textbf{forall} rules enclosing $\nu$, and let $\theta_\nu$ be the conjunction of the side conditions of Definition~\ref{ASMrules} for $\nu$, written with $\texttt{atomic}$ and with $\exists y\,(y=t)$ to express that the value of a term $t$ belongs to the primary part. For every $\nu$, and for all $\nu,\nu'$ with $f_\nu=f_{\nu'}$ of arity $r$, let
\begin{align*}
\alpha_\nu&=\multiset{(t_{\nu0},\ldots,t_{\nu r_\nu},\bar x_\nu)\mid\gamma_\nu\wedge\neg\big(f_\nu(t_{\nu1},\ldots,t_{\nu r_\nu})=t_{\nu0}\big)}_{\bar x_\nu},\\
\delta_\nu&=\multiset{(\texttt{true},\bar x_\nu)\mid\gamma_\nu\wedge\neg\theta_\nu}_{\bar x_\nu},\\
\beta_{\nu\nu'}&=\multiset{(\texttt{true},\bar x_\nu,\bar x'_{\nu'})\mid\gamma_\nu\wedge\gamma'_{\nu'}\wedge{\textstyle\bigwedge_{i=1}^{r}}(t_{\nu i}=t'_{\nu'i})\wedge\neg(t_{\nu0}=t'_{\nu'0})}_{\bar x_\nu,\bar x'_{\nu'}},
\end{align*}
where $\bar x'_{\nu'}$, $\gamma'_{\nu'}$ and $t'_{\nu'i}$ are obtained from $\bar x_{\nu'}$, $\gamma_{\nu'}$ and $t_{\nu'i}$ by renaming the variables apart from $\bar x_\nu$. $W_{\mathcal M}$ is the set of all witness terms $\alpha_\nu$, $\delta_\nu$ and $\beta_{\nu\nu'}$.
\end{definition}

The bound variables are included in the tuples, so that the terms in $W_{\mathcal M}$ are witness terms. The guards of the terms $\alpha_\nu$ exclude trivial updates by the conjunct $\neg(f_\nu(t_{\nu1},\ldots,t_{\nu r_\nu})=t_{\nu0})$, and the terms $\delta_\nu$ and $\beta_{\nu\nu'}$ detect undefined and inconsistent update sets, respectively.
\end{newtext}

\begin{theorem}[Plausibility]\label{Th:Plausibility}
Every parallel ASM $\cal M$ defines a parallel algorithm with the same vocabulary and background as $\cal M$.
\end{theorem}

\begin{proof}
Let $\mathcal{M}$ be a parallel ASM. We have to show that the four postulates (Sequential Time, Abstract State, Background and Parallel Bounded Exploration postulates) are satisfied. The Sequential Time and Background postulates are already built into the definition of an ASM. The same holds for the Abstract State postulate, and the preservation of isomorphisms is straightforward.

In what follows, we prove that $\mathcal{M}$ satisfies also the Bounded Exploration postulate \new{with bounded exploration witness $W_{\mathcal M}$ (Definition~\ref{def:witnessSet})}. \new{By Lemma~\ref{lem:udu-coincide}, it suffices to show that $\mathcal M$ has uniformly definable updates with respect to $W_{\mathcal M}$.}

\begin{newtext}
Let ${\bf S}$ be a state of computation corresponding to a state ${\bf I}$ of $\mathcal M$. An \emph{instance} of $\nu$ in ${\bf S}$ is a tuple $\bar c$ of elements of the primary part with ${\bf S}\models\gamma_\nu[\bar c]$; it yields the update
\[
\big((f_\nu,(\val_{{\bf S},\bar c}(t_{\nu1}),\ldots,\val_{{\bf S},\bar c}(t_{\nu r_\nu}))),\val_{{\bf S},\bar c}(t_{\nu0})\big),
\]
where $\val_{{\bf S},\bar c}$ denotes evaluation under the assignment of $\bar c$ to $\bar x_\nu$. By Definition~\ref{ASMrules}, $\Delta(r_{\mathcal M},{\bf S})$ is defined iff every instance satisfies the side conditions of the assignment rule, and then it is the set of the updates yielded by all instances. By Definition~\ref{ASM}, $\tau_{\mathcal M}({\bf I})-{\bf I}$ is the set of non-trivial updates in $\Delta(r_{\mathcal M},{\bf S})$ if $\Delta(r_{\mathcal M},{\bf S})$ is defined and consistent, and empty otherwise. Now $\Delta(r_{\mathcal M},{\bf S})$ is undefined iff $\val_{\bf S}(\delta_\nu)\neq\oslash$ for some $\nu$, and if it is defined, it is inconsistent iff $\val_{\bf S}(\beta_{\nu\nu'})\neq\oslash$ for some $\nu,\nu'$. Hence $\tau_{\mathcal M}({\bf I})-{\bf I}=\emptyset$ if one of these values is non-empty, and otherwise
\[
\begin{aligned}
\tau_{\mathcal M}({\bf I})-{\bf I}=\{((f_\nu,(b_1,\ldots,b_{r_\nu})),b_0)\mid\ &(b_0,\ldots,b_{r_\nu},\bar d)\in\val_{\bf S}(\alpha_\nu)\\
&\text{for some }\nu\text{ and some }\bar d\}.
\end{aligned}
\tag{$\ast$}
\]
In particular, $W_{\mathcal M}$ covers the updates of $\mathcal M$. Let $\mathit{Stop}$ be the disjunction of the equality-free sentences $\exists\bar z\,R_{\delta_\nu}(\bar z)$ and $\exists\bar z\,R_{\beta_{\nu\nu'}}(\bar z)$, which hold in ${\bf S}|_{W_{\mathcal M}}$ iff the corresponding value is non-empty. For every dynamic function symbol $f$ of arity $r$, the formula
\[
\psi_f(x_0,\ldots,x_r)=\neg\mathit{Stop}\wedge\bigvee_{\nu:\,f_\nu=f}\exists\bar u\,\exists s\,R_{\alpha_\nu}(x_0,\ldots,x_r,\bar u,s),
\]
with $\bar u$ of the length of $\bar x_\nu$, satisfies Definition~\ref{def:udu} by $(\ast)$.
\end{newtext}
\end{proof}

\section{Examples}
\label{sec:examples}

In this section we provide evidence that the general axiomatic description of parallel algorithms in Definition~\ref{def:parallelAlgo} captures the notion of deterministic, parallel algorithm that works synchronously on a fixed level of abstraction. We discuss some familiar approaches to parallelism in a way that is analogous to the discussion in the ASM thesis for parallel algorithms of Blass and Gurevich \cite{[BG03],[BG08]}, showing that they fit our (simplified) axiomatic characterization.

\begin{newtext}
By Theorem~\ref{Th:Plausibility}, an algorithm specified by a parallel ASM rule satisfies the Bounded Exploration postulate with the witness $W_{\mathcal M}$ of Definition~\ref{def:witnessSet}. The witness terms listed below describe the parts of the states that are explored by the algorithms.
\end{newtext}

\subsection{Circuits}

A popular model for parallel computing, which was shown to satisfy the postulates in the original ASM thesis for parallel algorithms  \cite{[BG03]}, is provided by the class of unbounded fan-in combinational Boolean circuit (see for instance \cite{[KR90]}). We show next that this kind of circuits also fit the new axiomatization proposed in this work.

A \emph{Boolean circuit of unbounded fan-in} with $n$ inputs $x_1, \ldots, x_n$ is a labeled acyclic digraph with a set of nodes $V$ (usually called gates), a set of edges $E$ and a labeling function $\lambda$ from $V$ to $\{x_1, \ldots, x_n\} \cup \{\land, \lor, \neg\}$ such that: $\lambda(v) \in \{x_1, \ldots, x_n\}$ implies that $v$ has in-degree $0$ and $\lambda(v) = \neg$ implies that $v$ has in-degree $1$. The in-degree of a node is called \emph{fan-in}. The fan-in is unbounded for nodes labeled with $\land$ or $\lor$.   

We consider the states of computation of the algorithm to include, apart from the requirements of the background postulate, the following functions:
\begin{itemize}
	\item A static bridge function $f_V$ which evaluates to $true$ if its argument is a node in $V$ and to $\texttt{false}$ otherwise.
	\item A static bridge function $f_E$ such that $f_E(v_1,v_2)$ evaluates to $\texttt{true}$ if there is an edge from $v_1$ to $v_2$ in $E$.
  \item A static function $f_\lambda$ in the primary part of the state such that $f_\lambda(v) = \lambda(v)$ if $v \in V$ and $f_\lambda(v) = \texttt{undef}$ otherwise. 
	\item A dynamic bridge function $\texttt{val}$ which assigns to each node $v \in V$ with $\lambda(v) = x_i$ for some $x_i$ in $\{x_1, \ldots, x_n\}$ the input value given to $x_i$, and to each node $v$ with $\lambda(v) \not\in \{x_1, \ldots, x_n\}$ the value it computes (which is $\texttt{undef}$ in the initial state and updated exactly once during the computation).
\end{itemize}

\begin{algorithm}[htp!]
\caption{Boolean Circuits of Unbounded Fan-in}
\label{ASM:Circuits}
\begin{algorithmic}[0]
\small
\FORALL {$x \textbf{ with } f_V(x)$}
   \IF {$\forall y (f_E(y,x) \rightarrow \texttt{val}(y) \neq \texttt{undef}) \wedge \texttt{val}(x) = \texttt{undef}$} 
	 \STATE {\bf par}
      \STATE {\bf if} $f_\lambda(x) = \neg$ {\bf then} $\texttt{val}(x) := \multiset{y \mid f_E(y,x) \land \texttt{val}(y) = \texttt{true}} = \oslash$ {\bf endif}
      \STATE {\bf if} $f_\lambda(x) = \lor$ {\bf then} $\texttt{val}(x) := \multiset{y \mid f_E(y,x) \land \texttt{val}(y) = \texttt{true}} \neq \oslash$ {\bf endif}
      \STATE {\bf if} $f_\lambda(x) = \land$ {\bf then} $\texttt{val}(x) := \multiset{y \mid f_E(y,x) \land \texttt{val}(y) = \texttt{false}} = \oslash$ {\bf endif}
	 \STATE {\bf endpar}	
   \ENDIF
\ENDFOR
\end{algorithmic}
\end{algorithm}

It is easy to see that the transition function described by Algorithm~\ref{ASM:Circuits}, together with the states described above, define a parallel algorithm that satisfies the Sequential Time, Abstract State, and Background postulates and computes the boolean function that corresponds to the input Boolean circuit. 
\new{By Theorem~\ref{Th:Plausibility}, it also satisfies the Bounded Exploration Postulate. The following witness terms describe the part of the states explored by the algorithm.}
 \begingroup
\allowdisplaybreaks

\begin{align*} 
\alpha_1 =&  \{\!\!\{ (\multiset{y \mid f_E(y,x) \land \texttt{val}(y) = \texttt{true}} = \oslash,x) \mid  f_V(x) \wedge  \\
&\forall y (f_E(y,x) \rightarrow \texttt{val}(y) \neq \texttt{undef}) \wedge \texttt{val}(x) = \texttt{undef} \wedge f_\lambda(x) = \neg\}\!\!\} \\ 
\alpha_2 =&  \{\!\!\{ (\multiset{y \mid f_E(y,x) \land \texttt{val}(y) = \texttt{true}} \neq \oslash,x) \mid  f_V(x) \wedge \\
&\forall y (f_E(y,x) \rightarrow \texttt{val}(y) \neq \texttt{undef}) \wedge \texttt{val}(x) = \texttt{undef} \wedge f_\lambda(x) = \lor \}\!\!\} \\ 
\alpha_3 =&  \{\!\!\{ (\multiset{y \mid f_E(y,x) \land \texttt{val}(y) = \texttt{false}} = \oslash,x) \mid  f_V(x) \wedge  \\
&\forall y (f_E(y,x) \rightarrow \texttt{val}(y) \neq \texttt{undef}) \wedge \texttt{val}(x) = \texttt{undef} \wedge f_\lambda(x) = \land \}\!\!\} \\
\alpha_4 =& \{\!\!\{ f_\lambda(x) = \neg \mid \forall y (f_E(y,x) \rightarrow \texttt{val}(y) \neq \texttt{undef}) \wedge \texttt{val}(x) = \texttt{undef} \wedge \\
& f_V(x)\}\!\!\} \\
\alpha_5 =& \{\!\!\{f_\lambda(x) = \lor \mid \forall y (f_E(y,x) \rightarrow \texttt{val}(y) \neq \texttt{undef}) \wedge \texttt{val}(x) = \texttt{undef} \wedge \\
& f_V(x)\}\!\!\} \\
\alpha_6 =& \{\!\!\{f_\lambda(x) = \land \mid \forall y (f_E(y,x) \rightarrow \texttt{val}(y) \neq \texttt{undef}) \wedge \texttt{val}(x) = \texttt{undef} \wedge \\ 
& f_V(x)\}\!\!\} \\
\alpha_7 =& \multiset{\forall y (f_E(y,x) \rightarrow \texttt{val}(y) \neq \texttt{undef}) \wedge \texttt{val}(x) = \texttt{undef} \mid  f_V(x)} \\
\alpha_8 =& \multiset{\texttt{true} \mid \texttt{true}}
\end{align*}
\endgroup
The set of witness terms additionally includes the terms obtained by replacing $f_V(x)$ by its negation in each of the witness terms $\alpha_1$,\ldots, $\alpha_7$. To simplify our presentation, we omit them here.  

\subsection{Parallel Random Access Machine}

The parallel random access machine (PRAM) model is a most idealized and powerful model of parallel computation. Following~\cite{[CP2003]}, we define a PRAM program as a sequence of random access machines (RAM) programs, $P = (\Pi_1, \Pi_2, \ldots, \Pi_q)$ one for each of $q$ RAMs. In turn, each RAM $i$ ($1 \leq i \leq q$) is a finite sequence $(\pi_{i1}, \pi_{i2}, \ldots, \pi_{im_i})$ of instructions of the kinds shown in Figure~2.6 in \cite{[CP2003]} (\texttt{READ}, \texttt{STORE}, \texttt{LOAD}, \texttt{JUMP}, etc), with arguments standing for the contents of registers (memory locations). Register $i$ is the accumulator of the RAM $i$, where the result of the current operation is stored. All registers, including the accumulators, are shared. That is, every RAM can both read and write all registers. Every RAM $i$ executes its own program $\Pi_i$. At each step, the RAM $i$ executes the instruction pointed by the program counter $\kappa_i$, reading and writing integer values on the registers as required by the instruction. There is also a set of input registers $I = (i_1, \ldots, i_m)$. The \texttt{HALT} instruction stops the computation by setting the program counter to $0$. Every semantically wrong instruction is considered as a \texttt{HALT} instruction. 
We assume that, apart from the requirements of the background postulate, the states of computation include the following functions:
\begin{itemize}
\item A unary bridge functions $P$ which is Boolean and static, and evaluate to $\texttt{true}$ only for the RAM programs $(1, \ldots, p)$.  
\item A unary static functions $I$ in the secondary part. If $j$ is an input register, $I(j)$ evaluates to its value. Otherwise, it evaluates to $\texttt{undef}$.
\item Binary bridge functions $\texttt{Inst}$, $\texttt{OpType}$ and $\texttt{OpVal}$ in the primary part which are static and map each program $i$ and line $l$ of $i$ to the instruction, the type of operand (either $j$, ${\uparrow}j$ or ${=}j$) and the operand value (an integer), respectively, that appears in the line $l$ of the program $i$. If $i$ is not a program or $l$ is not a line of $i$, then $\texttt{Inst}$, $\texttt{OpType}$ and $\texttt{OpVal}$ evaluate to \texttt{undef}.
\item A unary dynamic function $R$ from the positive integers to the integers which belongs to the secondary part and maps each register to its current value. 
\item A unary dynamic function $\kappa$ which belongs to the primary part and maps each program to the current value of its program counter.  
\item A unary bridge function $W$ which is dynamic and maps processors to pairs of the form (register, value). This is an auxiliary function which is used to collect the processors requests to update registers. It allows our ASM to detect whether more than one processor try to update the same register. In such cases, we use the convention that the processor with the smallest index prevails and has its value written in the register.
\item A nullary bridge function $\texttt{mode}$ which is interpreted by the values $0$ or $1$ depending on whether our ASM needs to execute a PRAM step or update the registers, respectively. 
\end{itemize}

\afterpage{
\begin{algorithm}[htp!]
\caption{Parallel Random Access Machine\footnotemark}
\label{ASM:PRAM}
\begin{algorithmic}[0]
\small
\STATE {{\bf if} $\texttt{mode} = 0$ {\bf then}}
\FORALL {$i \textbf{ with } P(i)$}
   \IF {$\kappa(i) \neq 0$} 
      \IF {$\texttt{Inst}(i, \kappa(i)) = \texttt{READ}$}
         \STATE {$\kappa(i) := \kappa(i) + 1$}
         \STATE {{\bf if} $\texttt{OpType}(i, \kappa(i)) = j$ {\bf then} $W(i) := (i, I(\texttt{OpVal}(i, \kappa(i))))$ {\bf end if}}
         \STATE {{\bf if} $\texttt{OpType}(i, \kappa(i)) = {\uparrow}j$ {\bf then}}
         \STATE {\qquad $W(i) := (i, I(R(\texttt{OpVal}(i, \kappa(i)))))$ {\bf end if}}
      \ENDIF
      \IF {$\texttt{Inst}(i, \kappa(i)) = \texttt{STORE}$}
         \STATE {$\kappa(i) := \kappa(i) + 1$}
         \STATE {{\bf if} $\texttt{OpType}(i, \kappa(i)) = j$ {\bf then} $W(i) := (\texttt{OpVal}(i, \kappa(i)), R(i))$ {\bf end if}}
         \STATE {{\bf if} $\texttt{OpType}(i, \kappa(i)) = {\uparrow}j$ {\bf then}}
         \STATE {\qquad $W(i) := (R(\texttt{OpVal}(i, \kappa(i))), R(i))$ {\bf end if}}
      \ENDIF
      \STATE $\vdots$
      \IF {$\texttt{Inst}(i, \kappa(i)) = \texttt{HALT}$}
         \STATE {$\kappa(i) := 0$}
      \ENDIF
   \ENDIF
\ENDFOR
\STATE {$\texttt{mode} := 1$}
\STATE {{\bf end if}}
\STATE {{\bf if} $\texttt{mode} = 1$ {\bf then}}
\FORALL {$i \textbf{ with } W(i) \neq \texttt{undef}$} 
   \IF {$\multiset{x \mid \texttt{first}(W(x)) = \texttt{first}(W(i)) \wedge x < i} = \oslash$}
      \STATE {$R(\texttt{first}(W(i))) := \texttt{second}(W(i))$}
   \ENDIF
   \STATE {$W(i) := \texttt{undef}$} 
\ENDFOR
\STATE {$\texttt{mode} := 0$}
\STATE {{\bf end if}}
\end{algorithmic}
\end{algorithm}
\footnotetext{We have intentionally omitted the par blocks so that the algorithm fits in one page. Sequences of instructions at a same level are implicitly assumed to be executed in parallel.}
}

The base set of the (finite) primary part of each state of our algorithm includes: a finite totally ordered set of RAM programs, a finite set of instructions, three operands, a finite set of operand values (integers which appear in the programs) and a finite totally ordered set of line numbers. We assume that the primary part also includes a successor function $+1$ defined on the set of line numbers, a constant $0$ which does not belong to the set of line numbers, the relation ``$<$'' defined in the set of RAM programs and constants $j$, ${\uparrow}j$ and ${=}j$ for the operands. The secondary part includes (in addition to the elements in the primary part) the set of integers. Also we assume that the secondary part includes functions for the standard arithmetic operations and relations among integers. In every initial state, if $i$ is a program, $\kappa(i)$ evaluates to the first line of $i$ and $W(i)$ evaluates to $\texttt{undef}$, and  $R(j)$ evaluates to $0$ if $j$ is a register (i.e., a positive integer) and to $\texttt{undef}$ otherwise. 

To simplify and shorten the exposition we assume that every instruction is semantically correct. That is, we assume that the programs do not contain semantically wrong instructions such as one that addresses Register $-14$. For the very same reason, in Algorithm~\ref{ASM:PRAM} we only show the cases corresponding to the $\texttt{READ}$, $\texttt{STORE}$ and $\texttt{HALT}$ instructions.  


\new{The transition function described by the ASM in Algorithm~\ref{ASM:PRAM} together with the states described above define a parallel algorithm; it satisfies the Bounded Exploration postulate by Theorem~\ref{Th:Plausibility}. The following witness terms describe the part of the states explored by the algorithm.} 
\begingroup
\allowdisplaybreaks

\begin{align*}
\alpha_1 =&  \{\!\!\{(\kappa(i) + 1, i) \mid \texttt{mode} = 0 \wedge P(i) \wedge \kappa(i) \neq 0 \wedge \texttt{Inst}(i, \kappa(i)) = \texttt{READ}\}\!\!\} \\
\alpha_2 =&  \{\!\!\{((i, I(\texttt{OpVal}(i, \kappa(i)))), i) \mid \texttt{mode} = 0 \wedge P(i) \wedge \kappa(i) \neq 0 \wedge \\
& \texttt{Inst}(i, \kappa(i)) = \texttt{READ} \wedge \texttt{OpType}(i, \kappa(i)) = j \}\!\!\} \\
\alpha_3 =&  \{\!\!\{((i, I(R(\texttt{OpVal}(i, \kappa(i))))) ,i) \mid \texttt{mode} = 0 \wedge P(i) \wedge \kappa(i) \neq 0 \wedge \\
& \texttt{Inst}(i, \kappa(i)) = \texttt{READ} \wedge \texttt{OpType}(i, \kappa(i)) = {\uparrow}j \}\!\!\} \\
\alpha_4 =&  \{\!\!\{(\kappa(i) + 1, i) \mid \texttt{mode} = 0 \wedge P(i) \wedge \kappa(i) \neq 0 \wedge \texttt{Inst}(i, \kappa(i)) = \texttt{STORE}\}\!\!\} \\
\alpha_5 =&  \{\!\!\{((\texttt{OpVal}(i, \kappa(i)), R(i)), i) \mid \texttt{mode} = 0 \wedge P(i) \wedge \kappa(i) \neq 0 \wedge \\
& \texttt{Inst}(i, \kappa(i)) = \texttt{STORE} \wedge \texttt{OpType}(i, \kappa(i)) = j \}\!\!\} \\
\alpha_6 =&  \{\!\!\{((R(\texttt{OpVal}(i, \kappa(i))), R(i)) ,i) \mid \texttt{mode} = 0 \wedge P(i) \wedge \kappa(i) \neq 0 \wedge \\
& \texttt{Inst}(i, \kappa(i)) = \texttt{STORE} \wedge \texttt{OpType}(i, \kappa(i)) = {\uparrow}j \}\!\!\} \\
\alpha_7 =&  \{\!\!\{\texttt{Inst}(i, \kappa(i)) = \texttt{READ} \mid \texttt{mode} = 0 \wedge P(i) \wedge \kappa(i) \neq 0 \}\!\!\}  \\
\alpha_8 =&  \{\!\!\{\texttt{Inst}(i, \kappa(i)) = \texttt{STORE} \mid \texttt{mode} = 0 \wedge P(i) \wedge \kappa(i) \neq 0 \}\!\!\}  \\
\alpha_{9} =&  \{\!\!\{\kappa(i) \neq 0 \mid P(i) \wedge \kappa(i) \neq 0 \}\!\!\}  \\
\vdots \\
\alpha_n =&  \{\!\!\{(0, i) \mid \texttt{mode} = 0 \wedge P(i) \wedge \kappa(i) \neq 0 \wedge \texttt{Inst}(i, \kappa(i)) = \texttt{HALT}\}\!\!\} \\
\alpha_{n+1} =&  \{\!\!\{\texttt{Inst}(i, \kappa(i)) = \texttt{HALT} \mid \texttt{mode} = 0 \wedge P(i) \wedge \kappa(i) \neq 0 \}\!\!\}  \\
\alpha_{n+2} =& \{\!\!\{1 \mid \texttt{mode} = 0\}\!\!\} \\ 
\alpha_{n+3} =&  \{\!\!\{\texttt{mode} = 0 \mid \texttt{true} \}\!\!\}  \\
\alpha_{n+4} =& \{\!\!\{(\texttt{second}(W(i)), \texttt{first}(W(i))) \mid \texttt{mode} = 1 \wedge W(i) \neq \texttt{undef} \wedge \\
& \multiset{x \mid \texttt{first}(W(x)) = \texttt{first}(W(i)) \wedge x < i} = \oslash \}\!\!\} \\
\alpha_{n+5} =& \{\!\!\{(\texttt{undef},i) \mid \texttt{mode} = 1 \wedge W(i) \neq \texttt{undef} \}\!\!\} \\
\alpha_{n+6} =& \{\!\!\{\multiset{x \mid \texttt{first}(W(x)) = \texttt{first}(W(i)) \wedge x < i} = \oslash \mid \texttt{mode} = 1 \wedge \\
&W(i) \neq \texttt{undef} \}\!\!\} \\
\alpha_{n+7} =& \{\!\!\{0 \mid \texttt{mode} = 1 \}\!\!\} \\
\alpha_{n+8} =& \{\!\!\{\texttt{mode} = 1 \mid \texttt{true}\}\!\!\} \\
\alpha_{n+9} =& \{\!\!\{\texttt{true} \mid \texttt{true}\}\!\!\} 
\end{align*}
\endgroup
The set of witness terms additionally includes the terms obtained by replacing $P(i)$ by its negation in each of the witness terms $\alpha_1$,\ldots,$\alpha_{n+1}$ and also the terms obtained by replacing $W(i) \neq \texttt{undef}$ by its negation in each of the witness terms $\alpha_{n+4}$,\ldots,$\alpha_{n+6}$.

\subsection{Alternating Turing Machines}

An {\em alternating Turing machine} is a variation of the nondeterministic Turing machine whose set of control states $Q$ is divided into four classes $Q_\exists$, $Q_\forall$, $Q_{\textit{acc}}$ and $Q_{\textit{rej}}$. Thus there are existential, universal, accepting and rejecting states. The notion of acceptance is defined by induction on the computation tree. The alternating Turing machine in a given configuration $c$ {\em accepts} iff $c$ is in a final accepting state, or $c$ is in an existential state and at least one of its children in the computation tree accepts, or $c$ is in an universal state, it has at least one child in the computation tree, and all its children in the computation tree accept. Note that an alternating machine all of whose non-final control states are existential, is essentially a nondeterministic Turing machine.  

We begin our description of alternating Turing machines as a parallel algorithm which satisfies our postulates, by describing the computation states of the algorithm. They include, in addition to the requirements of the Background postulate, the following:
\begin{itemize}
\item Unary bridge functions $\Gamma$, $Q_\exists$, $Q_\forall$, $Q_{\textit{acc}}$ and $Q_{\textit{rej}}$ which are Boolean valued and evaluate to $\texttt{true}$ only for those elements that are symbols in the tape alphabet (which includes a symbol ``$\_$'' for blank) and control states in the classes of existential, universal, accepting and rejecting states, respectively. 
\item A unary bridge function $\texttt{config}$ which maps nodes in the computation tree to configurations. Elements that do not correspond to nodes in the computation tree are mapped to $\texttt{undef}$. A \emph{configuration} is a string in $\Gamma^* \cdot Q \cdot \Gamma^*$ that specifies the content of the tape, the state, and the position of the head as follows. Let the tape contain symbols $w_n, w_{n+1}, \ldots, w_{m-1}, w_m$ where $w_i \in \Gamma$ is the symbol on the $i$-th position of the tape and $n < m$. Assume that the read/write  head is in position $n \leq j \leq m$ and $n$ and $m$ are such that for all $n' < n$, and $m' > m$, $w_{n} = w_{m} = \_$ (the blank symbol). If the machine is in control state $q$, we denote this configuration by $w_n \ldots w_{j-1}, q, w_j \ldots w_m$. 
\item A unary bridge function $\texttt{active}$ which evaluates to $\texttt{true}$ for an element $v$ if $\texttt{config}(v) \neq \texttt{undef}$ (i.e. $v$ is a node of the computation tree) and $v$ is computing (i.e., has not spawned sub-computations yet), to $\texttt{undef}$ if $\texttt{config}(v) = \texttt{undef}$, and to $\texttt{false}$ otherwise. 
\item A unary bridge function $\texttt{value}$ which evaluates to $\texttt{true}$ for an element $v$ if $v$ is an accepting node in the computation tree, to $\texttt{false}$ if $v$ is a rejecting node, and to $\texttt{undef}$ if the value of $v$ has not been established (by the algorithm) yet or $v$ is not a node.  
\item A bridge function $\delta$ that represents the transition relation $\delta \subseteq Q \times \Gamma \times Q \times \Gamma \times \{L,R\}$ of the alternating Turing Machine.
\item A binary, Boolean valued function $\texttt{parent}(x,y)$ which belongs to the primary (finite) part and evaluates to $\texttt{true}$ iff $x$ is the parent node of $y$ in the computation tree. 
\item Unary functions $\texttt{state}$ and $\texttt{read}$ which belong to the secondary part and map each configuration to the current state and symbol (in the position of the read/write head), respectively.
\item A function $\texttt{nextConf}$ of arity $4$ in the secondary part which maps a given configuration $c$, state $q$, symbol $a$ and movement $m$ to the configuration obtained by replacing the state in $c$ by $q$, the symbol in the current position of the read/write head by $a$, and updating the position of the head one cell to the left ($m = L$) or to the right ($m = R$).  
\end{itemize}

\afterpage{
\begin{algorithm}[htp!]
\caption{Alternating Turing Machines\footnotemark}
\label{ASM:AlternatingTM}
\begin{algorithmic}[0]
\small
\FORALL {$x \textbf{ with } \texttt{config}(x) \neq \texttt{undef}$} 
   \IF {$\texttt{active}(x)$} 
      \STATE {\bf if} {$Q_{\textit{acc}}(\texttt{state}(\texttt{config}(x))) \wedge \texttt{value}(x) = \texttt{undef}$} {\bf then}
      \STATE \quad $\texttt{value}(x) := \texttt{true}$ {\bf endif}
      \STATE {\bf if} {$Q_{\textit{rej}}(\texttt{state}(\texttt{config}(x))) \wedge \texttt{value}(x) = \texttt{undef}$} {\bf then}
      \STATE \quad $\texttt{value}(x) := \texttt{false}$ {\bf endif}
      \IF {$Q_{\exists}(\texttt{state}(\texttt{config}(x))) \vee Q_{\forall}(\texttt{state}(\texttt{config}(x)))$}
         \FORALL {$q,a,m \textbf{ with } \delta(\texttt{state}(\texttt{config}(x)),\texttt{read}(\texttt{config}(x)),q,a,m)$}
            \STATE $\textbf{import } c \textbf{ do}$
            \STATE \quad $\texttt{active}(c) := \texttt{true}$
            \STATE \quad $\texttt{config}(c) := \texttt{nextConfig}(x,q,a,m)$
            \STATE \quad $\texttt{value}(c) := \texttt{undef}$
            \STATE \quad $\texttt{parent}(x,c) := \texttt{true}$
            \STATE $\textbf{enddo}$
         \ENDFOR
         \STATE $\texttt{active}(x) := \texttt{false}$
      \ENDIF
   \ENDIF
   \IF {$\neg\texttt{active}(x)$} 
      \IF {$Q_{\exists}(\texttt{state}(\texttt{config}(x)))$}
         \STATE {\bf if} {$\exists y (\texttt{parent}(x,y) \wedge \texttt{value}(y) = \texttt{true})$} {\bf then} $\texttt{value}(x) := \texttt{true}$ {\bf endif}
         \STATE {\bf if} {$\forall y (\texttt{parent}(x,y) \wedge \texttt{value}(y) = \texttt{false})$} {\bf then} $\texttt{value}(x) := \texttt{false}$ {\bf endif}
      \ENDIF
      \IF {$Q_{\forall}(\texttt{state}(\texttt{config}(x)))$}
         \STATE {\bf if} {$\exists y (\texttt{parent}(x,y)) \wedge \forall y (\texttt{parent}(x,y) \wedge \texttt{value}(y) = \texttt{true})$} {\bf then}
         \STATE \quad $\texttt{value}(x) := \texttt{true}$ {\bf endif}
         \STATE {\bf if} {$\exists y (\texttt{parent}(x,y) \wedge \texttt{value}(y) = \texttt{false})$} {\bf then} $\texttt{value}(x) := \texttt{false}$ {\bf endif}
      \ENDIF
   \ENDIF
\ENDFOR
\end{algorithmic}
\end{algorithm}
\footnotetext{We have intentionally omitted the par blocks so that the algorithm fits in one page. Sequences of instructions at a same level are implicitly assumed to be executed in parallel.}}

The base set of the (finite) primary part of each meta-finite state includes a finite set of control states $Q = Q_\exists \cup Q_\forall \cup Q_{\textit{acc}} \cup Q_{\textit{rej}}$, a finite set of tape symbols $\Gamma$, a finite set of nodes of the computation tree and two distinguished elements $L$ and $R$.  The secondary part includes (in addition to the elements in the primary part) all possible configurations (finite strings) in $\Gamma^* \cdot Q \cdot \Gamma^*$. The initial states are those in which the computation tree is just a single node $v$ which is mapped via \texttt{config} to an initial configuration, $\texttt{value}(v)$ is undefined and $\texttt{active}(v)$ is true. An initial configuration is represented by a string $q_0 \cdot s$ where $s \in (\Gamma - \{\_\})^*$, that is, the control state is $q_0$, the read/write head points to the first position of $s$, and the tape contains $s$ preceded and followed by countably many blank cells.  Apart from $\texttt{active}$, $\texttt{config}$, $\texttt{value}$, $\texttt{parent}$ and $\texttt{reserve}$ (see the Background postulate) which are dynamic functions, all the others are static.

The transition function is defined by Algorithm~\ref{ASM:AlternatingTM}. Note that, during each step of the computation of an alternating Turing machine, new nodes are added to the computation tree. We handle this situation by importing new elements from the reserve following the schema in Section~\ref{Sec:theReserve} and using the construct ``$\textbf{import } x \textbf{ do } P$''. By extending the language with this construct, reserve elements will be chosen for all combinations of the values $x$ such that $x$ is in the scope of any rule having import as a sub-rule. This construct with the corresponding semantics will be made explicit and addressed in full detail elsewhere. In this example we simply assume that every state of computation has a function $\texttt{import}$ which uniformly (through all possible states of computation) maps tuples of values to reserve elements. 
\new{The part of the states explored by the algorithm is described by the following witness terms:}
\begingroup
\allowdisplaybreaks
\begin{align*} 
\alpha_1 =&  \{\!\!\{ (\texttt{true},x) \mid \texttt{config}(x) \neq \texttt{undef} \wedge \texttt{active}(x) \wedge\\
& Q_{acc}(\texttt{state}(\texttt{config}(x))) \wedge \texttt{value}(x) = \texttt{undef} \}\!\!\} \\ 
\alpha_2 =&  \{\!\!\{ (\texttt{false},x) \mid \texttt{config}(x) \neq \texttt{undef} \wedge \texttt{active}(x) \wedge\\
& Q_{rej}(\texttt{state}(\texttt{config}(x))) \wedge \texttt{value}(x) = \texttt{undef} \}\!\!\} \\ 
\alpha_3 =& \{\!\!\{ (\texttt{true}, \texttt{import}((x,q,a,m))) \mid \texttt{config}(x) \neq \texttt{undef} \wedge \texttt{active}(x) \wedge\\
& (Q_{\exists}(\texttt{state}(\texttt{config}(x))) \vee Q_{\forall}(\texttt{state}(\texttt{config}(x)))) \wedge \\
& \delta(\texttt{state}(\texttt{config}(x)),\texttt{read}(\texttt{config}(x)), q, a, m)\}\!\!\} \\ 
\alpha_4 =& \{\!\!\{ (\texttt{nextConfig}(x,q,a,m), \texttt{import}((x,q,a,m))) \mid \texttt{config}(x) \neq \texttt{undef} \wedge\\
& \texttt{active}(x) \wedge (Q_{\exists}(\texttt{state}(\texttt{config}(x))) \vee Q_{\forall}(\texttt{state}(\texttt{config}(x)))) \wedge \\
& \delta(\texttt{state}(\texttt{config}(x)),\texttt{read}(\texttt{config}(x)), q, a, m)\}\!\!\} \\ 
\alpha_5 =& \{\!\!\{ (\texttt{undef}, \texttt{import}((x,q,a,m))) \mid \texttt{config}(x) \neq \texttt{undef} \wedge \texttt{active}(x) \wedge\\
& (Q_{\exists}(\texttt{state}(\texttt{config}(x))) \vee Q_{\forall}(\texttt{state}(\texttt{config}(x)))) \wedge \\
& \delta(\texttt{state}(\texttt{config}(x)),\texttt{read}(\texttt{config}(x)), q, a, m)\}\!\!\} \\ 
\alpha_6 =& \{\!\!\{ (\texttt{true}, x, \texttt{import}((x,q,a,m))) \mid \texttt{config}(x) \neq \texttt{undef} \wedge \texttt{active}(x) \wedge\\
& (Q_{\exists}(\texttt{state}(\texttt{config}(x))) \vee Q_{\forall}(\texttt{state}(\texttt{config}(x)))) \wedge \\
& \delta(\texttt{state}(\texttt{config}(x)),\texttt{read}(\texttt{config}(x)), q, a, m)\}\!\!\} \\ 
\alpha_7 =& \{\!\!\{ (\texttt{false}, x) \mid \texttt{config}(x) \neq \texttt{undef} \wedge \texttt{active}(x) \wedge\\
& (Q_{\exists}(\texttt{state}(\texttt{config}(x))) \vee Q_{\forall}(\texttt{state}(\texttt{config}(x))))\}\!\!\} \\ 
\alpha_8 =& \{\!\!\{ (\texttt{true},x) \mid \texttt{config}(x) \neq \texttt{undef} \wedge \neg \texttt{active}(x) \wedge\\
& Q_{\exists}(\texttt{state}(\texttt{config}(x))) \wedge \exists y (\texttt{parent}(x,y) \wedge \texttt{value}(y) = \texttt{true})\}\!\!\} \\ 
\alpha_9 =& \{\!\!\{ (\texttt{false},x) \mid \texttt{config}(x) \neq \texttt{undef} \wedge \neg \texttt{active}(x) \wedge\\
& Q_{\exists}(\texttt{state}(\texttt{config}(x))) \wedge \forall y (\texttt{parent}(x,y) \wedge \texttt{value}(y) = \texttt{false})\}\!\!\} \\ 
\alpha_{10} =& \{\!\!\{ (\texttt{true},x) \mid \texttt{config}(x) \neq \texttt{undef} \wedge \neg \texttt{active}(x) \wedge\\
& Q_{\forall}(\texttt{state}(\texttt{config}(x))) \wedge \exists y (\texttt{parent}(x, y)) \wedge\\
& \forall y (\texttt{parent}(x,y) \wedge \texttt{value}(y) = \texttt{true})\}\!\!\} \\ 
\alpha_{11} =& \{\!\!\{ (\texttt{false},x) \mid \texttt{config}(x) \neq \texttt{undef} \wedge \neg \texttt{active}(x) \wedge\\
& Q_{\forall}(\texttt{state}(\texttt{config}(x))) \wedge \exists y (\texttt{parent}(x,y) \wedge \texttt{value}(y) = \texttt{false})\}\!\!\} \\
\alpha_{12}=&\{\!\!\{ \texttt{active}(x) \mid \texttt{config}(x) \neq \texttt{undef}\}\!\!\} \\ 
\alpha_{13}=&\{\!\!\{ \texttt{true} \mid \texttt{active}(x) \wedge \texttt{config}(x) \neq \texttt{undef}\}\!\!\} \\
\alpha_{14}=&\{\!\!\{ Q_{acc}(\texttt{state}(\texttt{config}(x))) \wedge \texttt{value}(x) = \texttt{undef} \mid \texttt{active}(x) \wedge \\
&\texttt{config}(x) \neq \texttt{undef}\}\!\!\} \\
\alpha_{15}=&\{\!\!\{ Q_{rej}(\texttt{state}(\texttt{config}(x))) \wedge \texttt{value}(x) = \texttt{undef} \mid \texttt{active}(x) \wedge \\
&\texttt{config}(x) \neq \texttt{undef}\}\!\!\} \\
\alpha_{16}=&\{\!\!\{ Q_{\exists}(\texttt{state}(\texttt{config}(x))) \vee Q_{\forall}(\texttt{state}(\texttt{config}(x)))) \mid \texttt{active}(x)\wedge\\
& \texttt{config}(x) \neq \texttt{undef}\}\!\!\} \\ 
\alpha_{17}=&\{\!\!\{ \neg \texttt{active}(x) \mid \texttt{config}(x) \neq \texttt{undef}\}\!\!\} \\
\alpha_{18}=&\{\!\!\{ Q_{\exists}(\texttt{state}(\texttt{config}(x)) \mid \neg \texttt{active}(x) \wedge \texttt{config}(x) \neq \texttt{undef}\}\!\!\} \\
\alpha_{19}=&\{\!\!\{ Q_{\forall}(\texttt{state}(\texttt{config}(x)) \mid \neg \texttt{active}(x) \wedge \texttt{config}(x) \neq \texttt{undef}\}\!\!\} \\
\alpha_{20}=&\{\!\!\{ \exists y (\texttt{parent}(x, y)) \wedge \texttt{value}(y) = \texttt{true}) \mid Q_{\exists}(\texttt{state}(\texttt{config}(x)) \wedge \\ 
&\neg \texttt{active}(x) \wedge \texttt{config}(x) \neq \texttt{undef}\}\!\!\} \\
\alpha_{21}=&\{\!\!\{ \forall y (\texttt{parent}(x, y)) \wedge \texttt{value}(y) = \texttt{false})\mid Q_{\exists}(\texttt{state}(\texttt{config}(x)) \wedge \\
&\neg \texttt{active}(x) \wedge \texttt{config}(x) \neq \texttt{undef}\}\!\!\} \\
\alpha_{22}=&\{\!\!\{ \exists y (\texttt{parent}(x, y)) \wedge \forall y (\texttt{parent}(x,y) \wedge \texttt{value}(y) = \texttt{true} \mid \\
&Q_{\forall}(\texttt{state}(\texttt{config}(x)) \wedge \neg \texttt{active}(x) \wedge \texttt{config}(x) \neq \texttt{undef}\}\!\!\} \\
\alpha_{23}=&\{\!\!\{ \exists y (\texttt{parent}(x, y) \wedge \texttt{value}(y) = \texttt{false}) \mid \\
&Q_{\forall}(\texttt{state}(\texttt{config}(x)) \wedge \neg \texttt{active}(x) \wedge \texttt{config}(x) \neq \texttt{undef}\}\!\!\} \\
\alpha_{24}=&\{\!\!\{ \texttt{true} \mid \texttt{true}\}\!\!\} 
\end{align*}
\endgroup
The set of witness terms additionally includes the terms $\alpha_3'$--$\alpha_6'$ obtained by replacing $\delta(state(config(x)), read(config(x)), q, a, m)$ by its negation in  $\alpha_3$--$\alpha_6$  and also terms obtained by replacing $config(x) \neq undef$ by its negation in each of the witness terms $\alpha_1$,\ldots,$\alpha_{23}$ and $\alpha_3'$--$\alpha_6'$.  

\subsection{First-Order Logic}

Another model for parallel computing, which was shown to satisfy the postulates in the original ASM thesis for parallel algorithms~\cite{[BG03]}, is  the evaluation of a given first-order sentence $\Phi$ on a finite input. To simplify the presentation, we consider here only structures of purely relational vocabulary. 

Let $\{D,R_1,\ldots,R_n\}$ be the finite input structure in which $\Phi$ is to be evaluated, where $D$ is the domain and $R_1,\ldots,R_n$ are relations over $D$. We consider the state of computations to include, in addition to the requirements of the Background postulate, the following functions:

\begin{itemize}
\item A unary relation $\texttt{subForm}$, which belongs to the primary part and is formed by the set 
\[\{(\varphi(\bar{x}),(\bar{a})) \mid \varphi(\bar{x}) \; \text{is a subformula of} \; \Phi \; \text{or} \; \Phi \; \text{itself, and} \; |\bar{a}| \in D^{|\textit{free}(\varphi)|}\}.\] 
\item A unary, Boolean valued bridge function $\texttt{eval}$, such that \[\texttt{eval}((R_i(\bar{x}),(\bar{a}))) = \texttt{true} \; \text{iff} \; R_i(\bar{a})=\texttt{true}.\]
\item A binary relation $\texttt{superForm}$ which belongs to the primary part and is formed by the set of pairs of the form $((\varphi(x_1,\ldots,x_n),(a_1,\ldots,a_n)),(\psi(x_1,$\\
$\ldots,x_n,x_{n+1},\ldots,x_{n+m}),(a_1,\ldots,a_n,$
$a_{n+1},\ldots,a_{n+m})))$ such that $n, m \ge 0$ and $(\psi(x_1,\ldots,x_n,x_{n+1},\ldots,x_{n+m}),(a_1,\ldots,a_n,a_{n+1},\ldots,a_{n+m}))$ is a formula corresponding to a first-level node in the syntactic tree of $\varphi(x_1,\ldots,x_n)$ and $(a_1,\ldots,a_n,a_{n+1},\ldots,a_{n+m})\in D^{n+m}$. 
\item A unary, dynamic bridge function $\texttt{truthVal}$ which maps a formula/tuple pair to its corresponding truth value.
\item A unary, Boolean valued function $\texttt{Atomic}$ which belongs to the primary (finite) part and evaluates to $\texttt{true}$ iff its argument corresponds to an atomic formula.
\item A unary, static bridge function $\texttt{mainConnect}$ which belongs to the primary part and maps a formula to its main connective $\land$, $\lor$ or $\neg$. 
\item A unary, static function $\texttt{mainQuant}$ which belongs to the primary part and maps a formula to its main quantifier $\exists$ or $\forall$.
 \end{itemize}

Thus, the base set of the primary part of each meta-finite state includes:
\begin{itemize}
\item A Boolean valued function $D(a)$ which evaluates to $\texttt{true}$ if $a$ is an element in the domain $D$. 
\item Boolean valued functions $R_1,\ldots,R_n$ representing the corresponding relations over $D$. 
\item The set $U = \{(a_1,\ldots,a_n) \mid n \le |\textit{Var}(\Phi)| \land D(a_1),\ldots,D(a_n)\}$, where $\textit{Var}(\Phi)$ denotes the set of all variables occurring in $\Phi$.
\item The set $A = \{\neg,\land,\lor,\exists,\forall,(,)\} \cup \{x_i \mid x_i \in \textit{Var}(\Phi)\}$.
\end{itemize}

\begin{algorithm}[htp!]
\caption{First Order Logic}
\label{ASM:FOL}
\begin{algorithmic}[0]
\small
\FORALL {$x \textbf{ with }  \texttt{subForm}(x)$} 
\IF{$\texttt{truthVal}(x) = \texttt{undef}$} 
\IF{$\texttt{Atomic}(x)$}
\STATE{$\quad\;\texttt{truthVal}(x) := \texttt{eval}(x)$}
\ENDIF
\IF{$\neg\texttt{Atomic}(x)$}
\IF{$\forall x' (\texttt{superForm}(x,x') \Rightarrow \texttt{truthVal}(x') \neq \texttt{undef})$}
\IF{$\texttt{mainConnect}(x) =\land$}
\STATE{$\texttt{truthVal}(x) := $}
\STATE{$\qquad\qquad\multiset{y \mid \texttt{superForm}(x,y) \land \texttt{truthVal}(y)=\texttt{false}}_x = \oslash$}
\ENDIF
\IF{$\texttt{mainConnect}(x) =\lor$}
\STATE{$\texttt{truthVal}(x) := \multiset{y \mid \texttt{superForm}(x,y) \land \texttt{truthVal}(y)=\texttt{true}}_x \neq \oslash$}
\ENDIF
\IF{$\texttt{mainConnect}(x) =\neg$}
\STATE{$\texttt{truthVal}(x) := \multiset{y \mid \texttt{superForm}(x,y) \land \texttt{truthVal}(y)= \texttt{true}}_x \neq \oslash$}
\ENDIF
\IF{$\texttt{mainQuant}(x) = \exists$}
\STATE{$\texttt{truthVal}(x) := \multiset{y \mid \texttt{superForm}(x,y) \land \texttt{truthVal}(y) =\texttt{true}}_{x} \neq \oslash$}
\ENDIF
\IF{$\texttt{mainQuant}(x) = \forall$}
\STATE{$\texttt{truthVal}(x) := \multiset{y \mid \texttt{superForm}(x,y) \land \texttt{truthVal}(y) \neq \texttt{true}}_{x} = \oslash$}
\ENDIF
\ENDIF
\ENDIF
\ENDIF
\ENDFOR
\end{algorithmic}
\end{algorithm}

The base set of the secondary part includes only the elements specified by the Background postulate. The initial states are those in which $\texttt{truthVal}(x)=\texttt{undef}$ for all $x$ such that $\texttt{subForm}(x)$. The transition function described by the ASM in Algorithm~\ref{ASM:FOL}, together with the states described above, show that the first-order logic as a parallel model of computation satisfies the Sequential Time, Abstract State, and Background postulates. 
\new{By Theorem~\ref{Th:Plausibility}, it also satisfies the Bounded Exploration postulate. The following witness terms describe the part of the states explored by the algorithm.}
\begingroup
\allowdisplaybreaks
\begin{align*}
\alpha_1 = &\multiset{(\texttt{eval}(x),x) \mid \texttt{Atomic}(x) \wedge \texttt{truthVal}(x)=\texttt{undef} \wedge \texttt{subForm}(x)}\\
\alpha_2 = &\multiset{\texttt{Atomic}(x) \mid \texttt{truthVal}(x)=\texttt{undef} \wedge \texttt{subForm}(x)}\\
\alpha_3 = &\multiset{\texttt{truthVal}(x)=\texttt{undef} \mid \texttt{subForm}(x)}\\
\alpha_{4}= &\{\!\!\{(\multiset{y \mid \texttt{superForm}(x,y) \land \texttt{truthVal}(y) =\texttt{false}}_{x}=\oslash,x)\mid \\
&\forall x'(\texttt{superForm}(x,x') \Rightarrow \texttt{truthVal}(x') \neq \texttt{undef}) \wedge  \neg\texttt{Atomic}(x) \wedge \\
&\texttt{truthVal}(x)=\texttt{undef} \wedge \texttt{subForm}(x)\wedge\texttt{mainConnect}(x)=\land\}\!\!\}\\
\alpha_{5}= &\{\!\!\{\texttt{mainConnect}(x)=\land \mid \forall x' (\texttt{superForm}(x,x') \Rightarrow \texttt{truthVal}(x') \neq  \\
&\texttt{undef})\wedge \neg\texttt{Atomic}(x) \wedge \texttt{truthVal}(x)=\texttt{undef} \wedge \texttt{subForm}(x)\}\!\!\}\\
\alpha_{6}= &\{\!\!\{(\multiset{y \mid \texttt{superForm}(x,y) \land \texttt{truthVal}(y) =\texttt{true}}_{x}\neq\oslash,x)\mid \\
&\forall x' (\texttt{superForm}(x,x') \Rightarrow \texttt{truthVal}(x') \neq \texttt{undef}) \wedge \neg\texttt{Atomic}(x) \wedge \\
&\texttt{truthVal}(x)=\texttt{undef} \wedge \texttt{subForm}(x)\wedge\texttt{mainConnect}(x)=\lor\}\!\!\}\\
\alpha_{7}= &\{\!\!\{\texttt{mainConnect}(x)=\lor \mid \forall x' (\texttt{superForm}(x,x') \Rightarrow \texttt{truthVal}(x') \neq\\
& \texttt{undef}) \wedge \neg\texttt{Atomic}(x) \wedge \texttt{truthVal}(x)=\texttt{undef} \wedge \texttt{subForm}(x)\}\!\!\}\\
\alpha_{8}= &\{\!\!\{(\multiset{y \mid \texttt{superForm}(x,y) \land \texttt{truthVal}(y) =\texttt{true}}_{x}\neq\oslash,x)\mid \\
&\forall x' (\texttt{superForm}(x,x') \Rightarrow \texttt{truthVal}(x') \neq \texttt{undef}) \wedge \neg\texttt{Atomic}(x) \wedge \\
&\texttt{truthVal}(x)=\texttt{undef} \wedge \texttt{subForm}(x) \wedge \texttt{mainConnect}(x)=\neg\}\!\!\}\\
\alpha_{9}= &\{\!\!\{\texttt{mainConnect}(x)=\neg \mid \forall x' (\texttt{superForm}(x,x') \Rightarrow \texttt{truthVal}(x') \neq \\
&\texttt{undef})\wedge \neg\texttt{Atomic}(x) \wedge \texttt{truthVal}(x)=\texttt{undef} \wedge \texttt{subForm}(x)\}\!\!\}\\
\alpha_{10}= &\{\!\!\{(\multiset{y \mid \texttt{superForm}(x,y) \land \texttt{truthVal}(y) =\texttt{true}}_{x}\neq\oslash,x)\mid \\
&\forall x' (\texttt{superForm}(x,x') \Rightarrow \texttt{truthVal}(x') \neq \texttt{undef}) \wedge \neg\texttt{Atomic}(x) \wedge \\
& \texttt{truthVal}(x)=\texttt{undef} \wedge \texttt{subForm}(x) \wedge \texttt{mainQuant}(x)=\exists \}\!\!\}\\
\alpha_{11}= &\{\!\!\{\texttt{mainQuant}(x)=\exists \mid \forall x' (\texttt{superForm}(x,x') \Rightarrow \texttt{truthVal}(x') \neq \texttt{undef}) \\
&\wedge \neg\texttt{Atomic}(x) \wedge \texttt{truthVal}(x)=\texttt{undef} \wedge \texttt{subForm}(x)\}\!\!\}\\
\alpha_{12}= &\{\!\!\{(\multiset{y \mid \texttt{superForm}(x,y) \land \texttt{truthVal}(y) =\texttt{true}}_{x}\neq\oslash,x)\mid \\
&\forall x' (\texttt{superForm}(x,x') \Rightarrow \texttt{truthVal}(x') \neq \texttt{undef}) \wedge \neg\texttt{Atomic}(x) \wedge \\
& \texttt{truthVal}(x)=\texttt{undef} \wedge \texttt{subForm}(x) \wedge \texttt{mainQuant}(x)=\forall\}\!\!\}\\
\alpha_{13}= &\{\!\!\{\texttt{mainQuant}(x)=\forall \mid \forall x' (\texttt{superForm}(x,x') \Rightarrow \texttt{truthVal}(x') \neq \texttt{undef})\\
&\wedge \neg\texttt{Atomic}(x) \wedge \texttt{truthVal}(x)=\texttt{undef} \wedge \texttt{subForm}(x)\}\!\!\}\\
\alpha_{14}= &\{\!\!\{\forall x' (\texttt{superForm}(x,x') \Rightarrow \texttt{truthVal}(x') \neq \texttt{undef}) \mid \neg\texttt{Atomic}(x) \wedge \\
& \texttt{truthVal}(x)=\texttt{undef} \wedge \texttt{subForm}(x)\}\!\!\}\\
\alpha_{15}= &\multiset{\neg\texttt{Atomic}(x) \mid \texttt{truthVal}(x)=\texttt{undef} \wedge \texttt{subForm}(x)}\\
\alpha_{16}= &\multiset{\texttt{truthVal}(x)=\texttt{undef} \mid \texttt{subForm}(x)}\\
\alpha_{17}= &\multiset{\texttt{true} \mid \texttt{true}}
\end{align*}
\endgroup

The set of witness terms additionally includes the terms $\alpha_1'-\alpha_{16}'$ obtained by replacing $\texttt{subForm}(x)$ by its negation in each of the witness terms $\alpha_1-\alpha_{16}$.

\subsubsection{Breadth-First Search with MapReduce.}

We next show that another quite common and general approach to parallelism such as MapReduce, fits our description. More specifically, we show through an example how the MapReduce approach (developed at Google~\cite{[Dean08]}) can be simulated by a parallel ASM which satisfies our postulates. We note that other familiar approaches to parallelism such as parallel random access machines, circuits, alternating Turing machines, first-order and fixed-point logic, can also be shown to fit our axiomatic characterization of parallel algorithms. But given the space restriction we leave this task as an exercise for the reader.

The core functions of MapReduce consist of \emph{map}, \emph{shuffle}, and \emph{reduce} functions, which process $(\text{key},\text{value})$ pairs as follows~\cite{[Dean07]}:
\begin{itemize}
\item \emph{Map function} takes $(\text{key},\text{value})$ pairs as input and produces as output a set of intermediate
$(\text{key},\text{value})$ pairs.
\item \emph{Shuffle function} takes the output set of intermediate $(\text{key},\text{value})$ pairs of the map function and produces a set of pairs of unique intermediate key and set of values associated with it.
\item \emph{Reduce function} takes the sorted output of the shuffle and sort function and performs reduction or merge operation on these sets of values to produce zero or one output value per unique key.
\end{itemize}

 Let us consider an undirected graph with a set of nodes $V$, a distinguished source node $s$ and an adjacency list for each node in $V$. The parallel breadth-first search algorithm (BFS) explores the nodes of the graph reachable from $s$. To keep track of progress, BFS colours every node \emph{white}, \emph{grey} or \emph{black}. Before the start of the traversal, the source node $s$ is coloured \emph{grey}. Each other node is coloured \emph{white}. Nodes become \emph{grey}, once they are visited, and later \emph{black}, when all its adjacent nodes have been visited. The process continues until there are no more grey nodes to process in the graph.
We simulate via the following ASM rule the transition function corresponding to an implementation of the BFS algorithm in the MapReduce framework:
\afterpage{
\begin{algorithm}[htp!]
\caption{Breadth-First Search with MapReduce\footnotemark}
\label{ASM:MapReduce}
\begin{algorithmic}[0]
\small
\IF{$\textit{phase}=\textit{map}$} 
\FORALL{$x \; { \bf with } \; f_V(x)$} 
 \IF{$ \neg(\textit{colour}(x)=\mathit{grey})$}
            \STATE{$\textit{mapout}(x, \textit{colour}(x)) := \texttt{true}$}
	\ENDIF					
  \IF {$\textit{colour}(x)=\textit{grey}$}
  \STATE{${\bf for all } \; y \; { \bf with} \; \textit{listedIn}(y, \textit{neighb}(x)) \; { \bf do}$}
\STATE{$\quad \textit{mapout}(y, \textit{grey}) := \texttt{true}$}
\STATE{${\bf end \; for}$}
\STATE{$\textit{mapout}(x, \textit{black}) := \texttt{true}$}
\ENDIF
\ENDFOR
\STATE{$\textit{phase}:=\textit{shuffle}$}
\ENDIF
\IF{$\textit{phase}=\textit{shuffle}$} 
\STATE{${\bf for \; all} \;  x \; {\bf with} \; \exists{y}(\textit{mapout}(x,y)) \; {\bf do}$} 
\STATE{$\quad \textit{valuesOf}(x):=\multiset{y\mid{\textit{mapout}(x,y)}}_{x}$}
\STATE{$\quad {\bf for \; all} \; y \; {\bf with} \; \textit{mapout}(x,y) \; {\bf do}$}
 \STATE{$\quad\quad\textit{mapout}(x,y):=\texttt{false}$}
\STATE{$\quad {\bf end \; for}$}
\STATE{${\bf end \; for}$}
 \STATE{$\textit{phase}:=\textit{reduce}$}
\ENDIF
\IF{$\textit{phase}=\textit{reduce}$} 
\STATE{${\bf for \; all} \; x \; {\bf with} \; \exists{y}(\textit{valuesOf}(x)=y) \; {\bf do}$}
\STATE{$\quad colour(x):=darkest(\textit{valuesOf}(x))$}
\STATE{$\quad valuesOf(x):=\oslash$}
\STATE{${\bf end \; for}$}
\STATE{$phase:=\textit{map}$}
\ENDIF
\end{algorithmic}
\end{algorithm}
\footnotetext{We have intentionally omitted the par blocks so that the algorithm fits in one page. Sequences of instructions at a same level are implicitly assumed to be executed in parallel.}}

We consider the (metafinite) states of computation to include, in addition to the requirements of the Background postulate, the following functions:
\begin{itemize}
\item Static, nullary functions (constants) $\textit{white}$, $\textit{grey}$, $\textit{black}$, $\textit{map}$, $\textit{shuffle}$ and $\mathit{reduce}$ in the primary part which evaluate to pairwise different elements.
\item A static, unary, Boolean valued function $f_V$ in the primary part such that $f_V(x) = \texttt{true}$ iff $x$ is a node in $V$.
\item A static, unary, function $\textit{neighb}$ in the primary part such that $\textit{neighb}(x) = y$ iff $y$ is the adjacency list of the node $x$.
\item A static, Boolean valued function  $\textit{listedIn}$ in the primary part such that $\textit{listedIn}(x,y)$ iff $x$ is in the adjacency list $y$.
\item A dynamic, unary function $\textit{colour}$ in the primary part which maps each node to a colour in $\{\textit{white}$, $\textit{grey}$, $\textit{black}\}$.
\item A dynamic, nullary function $\textit{phase}$ in the primary part which evaluates to either $\textit{map}$, $\textit{shuffle}$ or $\mathit{reduce}$.
\item A dynamic, bridge, Boolean valued function $\textit{mapout}(x,y)$ which evaluates to $\texttt{true}$ iff the node $x$ is mapped out by the algorithm to the colour $y$.
\item A dynamic, bridge, unary function $\textit{valuesOf}$ which maps each node to a multiset with underlying set $\{\textit{white}$, $\textit{grey}$, $\textit{black}\}$.
\item A static, Boolean valued function $\textit{listedIn}(x,y)$ in the secondary part which evaluates to $\texttt{true}$ iff $\textrm{Mult}(x,y) \geq 1$.
\item A static, unary function $\textit{darkest}$ in the secondary part which evaluates to the darkest colour in a multiset with underlying set $\{\textit{white}$, $\textit{grey}$, $\textit{black}\}$.
\end{itemize}
Thus, the base set of the (finite) primary part of a state of computation contains the nodes in $V$, the constants $\textit{white}$, $\textit{grey}$, $\textit{black}$, $\textit{map}$, $\textit{shuffle}$, $\mathit{reduce}$ and the adjacency list of every node $n$ in $V$. The base set of the secondary part includes only the elements specified by the Background postulate. The initial states are those in which the following conditions hold: $\mathit{phase} = \mathit{map}$; there is exactly one node $s$ (the source node) in $V$ such that $\textit{colour}(s) = \textit{grey}$; for all nodes $n_i$ in $V$ with $n_i \neq s$, it holds that $\textit{colour}(n_i) = \textit{white}$; $\mathit{mapout}(x,y) = \texttt{false}$ for every pair of elements in the primary part; and $\textit{valuesOf}(x)=\oslash$ for every element in the primary part. We consider every multiset $x$ with underlying set $\{\textit{white}$, $\textit{grey}$, $\textit{black}\}$ as part of the base set of every state of the algorithm. Thus they are atomic elements in every computation state.


\new{The following witness terms describe the part of the states explored by the parallel algorithm presented in this section; superfluous conjunctions with $\texttt{true}$ are omitted.}

\begingroup
\allowdisplaybreaks
\begin{align*}
\alpha_1 = &\multiset{(\texttt{true},x,\mathit{black}) \mid \mathit{colour}(x)= \textit{grey} \wedge f_V(x) \wedge \mathit{phase} = \mathit{map}}\\
\alpha_2 = &\{\!\!\{(\texttt{true},y,\mathit{grey}) \mid \mathit{listedIn}(y,\mathit{neighb}(x)) \wedge \mathit{colour}(x)= \textit{grey} \wedge f_V(x) \wedge \\
&\mathit{phase} = \mathit{map}\}\!\!\}\\
\alpha_3 = &\{\!\!\{(\texttt{true},y,\mathit{grey}) \mid \neg\mathit{listedIn}(y,\mathit{neighb}(x)) \wedge \mathit{colour}(x)= \textit{grey} \wedge f_V(x) \wedge \\
&\mathit{phase} = \mathit{map}\}\!\!\}\\
\alpha_4 = &\multiset{(\texttt{true},x,\mathit{colour}(x)) \mid \neg(\mathit{colour}(x) = \mathit{grey}) \wedge f_V(x) \wedge \mathit{phase} = \mathit{map}}\\
\alpha_5 = &\multiset{\neg(\mathit{colour}(x) = \mathit{grey}) \mid f_V(x) \wedge \mathit{phase} = \mathit{map}}\\
\alpha_6 = &\multiset{\mathit{colour}(x) = \mathit{grey} \mid f_V(x) \wedge \mathit{phase} = \mathit{map}}\\
\alpha_7 = &\multiset{\mathit{shuffle} \mid \mathit{phase} = \mathit{map}}\\
\alpha_8 = &\multiset{\mathit{map} \mid \texttt{true}}\\
\alpha_9 = &\multiset{(\multiset{y \mid \mathit{mapout}(x, y)}_x, x) \mid \exists y (\mathit{mapout}(x,y)) \wedge \mathit{phase} = \mathit{shuffle}}\\
\alpha_{10} = &\multiset{(\texttt{false}, x, y) \mid \mathit{mapout}(x,y) \wedge \exists y (\mathit{mapout}(x,y)) \wedge \mathit{phase} = \mathit{suffle}}\\
\alpha_{11} = &\multiset{\mathit{shuffle} \mid \texttt{true}}\\
\alpha_{12} = &\multiset{\mathit{reduce} \mid \mathit{phase} =\mathit{shuffle}}\\
\alpha_{13} = &\multiset{(\mathit{darkest}(\mathit{valuesOf}(x)), x) \mid \exists y (\mathit{valuesOf}(x) = y) \wedge \mathit{phase} = \mathit{reduce}}\\
\alpha_{14} = &\multiset{(\oslash, x) \mid \exists y (\mathit{valuesOf}(x) = y) \wedge \mathit{phase} = \mathit{reduce}}\\
\alpha_{15} = &\multiset{\mathit{reduce} \mid \texttt{true}}\\
\alpha_{16} = &\multiset{\mathit{map} \mid \mathit{phase} = \mathit{reduce}}\\
\alpha_{17} = &\multiset{\texttt{true} \mid \texttt{true}}
\end{align*}
\endgroup

The set of witness terms additionally includes the terms $\alpha_1'-\alpha_6'$ obtained by replacing $f_V(x)$ by its negation in $\alpha_1-\alpha_6$, the term $\alpha_9'$ obtained by replacing $\exists y (\mathit{mapout}(x,y))$ by its negation in $\alpha_9$ and also the terms $\alpha_{13}'$, $\alpha_{14}'$ obtained by replacing $\exists y (\mathit{valuesOf}(x) = y)$ by its negation in each of the terms $\alpha_{13}, \alpha_{14}$.

\section{The Characterization Theorem}
\label{sec:characterizationthm}

This section is devoted to prove the key characterization theorem of our parallel ASM thesis, which states that for every parallel algorithm there is a behaviourally equivalent parallel ASM. 
\new{The proof uses the formulae $\psi_f$ that define the update sets in the critical structures (Definition~\ref{def:udu}). We translate them into Boolean terms and use these terms as guards of \textbf{forall} rules whose variables range over the witnessing assignments of the witness terms. Unlike the proof for sequential algorithms in~\cite{[Gurevich00]}, the proof does not need a case distinction over finitely many representative states.} 


\begin{newtext}
Throughout this section we fix a parallel algorithm $A$ of vocabulary $\Sigma=(\Sigma_1,\Sigma_2,\Sigma_F)$ with background class ${\cal K}$, and a bounded exploration witness $W=\{\alpha_1,\ldots,\alpha_m\}$ of $A$ with $\alpha_i=\multiset{(t_{i0},\ldots,t_{in_i})\mid\varphi_i}_{x_1,\ldots,x_{k_i}}$; $\psi_f$ denotes the formulae of Definition~\ref{def:udu}. As in~\cite{[Gurevich00]}, we do not consider the import of reserve elements (Definition~\ref{Def:theReserve} and the subsequent definition): for every state ${\bf I}$ of $A$, $\tau_A({\bf I})$ has the same primary part as ${\bf I}$. Consequently, every update in $\tau_A({\bf I})-{\bf I}$ satisfies the side conditions of the assignment rule in Definition~\ref{ASMrules}.
\end{newtext}


\subsection{Types}

We can now restrict ourselves to consider the properties of tuples which are definable in a given logic over \emph{finite relational structures}.  For this, we use the model-theoretic concept of type.  

\begin{definition}\label{def2-1}

Let $\cal L$ be a logic, let ${\bf A}$ be a relational structure of vocabulary
$\Sigma$, and let $\bar{a} = (a_1,  \ldots,  a_k)$ be  a $k$-tuple
over ${\bf A}$. The $\cal L$-{\em  type} of $\bar{a}$ in $\bf A$,  denoted $tp_{\bf
A}^{\cal L}(\bar{a})$, is the set ${\cal L}[\Sigma]$ of formulas in ${\cal L}$ of vocabulary $\Sigma$ with free
variables among $\{x_1, \ldots, x_k\}$ which are satisfied in ${\bf A}$ 
by any variable assignments assigning for $1 \leq i \leq k$ the $i$-th component of
$\bar{a}$ to the variable $x_i$, i.e.
\[ tp_{\bf A}^{\cal L}(\bar{a}) = \{\varphi \in {\cal L}[\Sigma]: \mathit{free}(\varphi) \subseteq
\{x_1, \ldots, x_k\} \textrm{ and } {\bf A} \models \varphi[a_1, \ldots, a_k] \}. \]

\end{definition}

Note that,  the ${\cal L}$-type of a given tuple $\bar{a}$ over a relational
structure ${\bf A}$, includes not only the properties of all sub-tuples of 
$\bar{a}$, but also the set of all sentences in ${\cal L}$ which are true 
when evaluated on ${\bf A}$.

In particular,  we are interested in the properties of tuples which are definable in first-order logic with and without equality (denoted $\mathrm{FO}$ and $\mathrm{FO}_{wo=}$, respectively), i.e., we are interested in $\mathrm{FO}$-types and $\mathrm{FO}_{wo=}$-types, respectively.  $\mathrm{FO}$-types are also known as isomorphism types since every tuple can be characterized up to isomorphism by its $\mathrm{FO}$-type. $\mathrm{FO}_{wo=}$-types correspond to a (weaker) type of equivalence relation among tuples (see \cite{[CDJ96]}). Instead of the partial isomorphism condition used in the Ehrenfeucht-Fra\"{\i}ss\'e characterization of $\mathrm{FO}$ equivalence, the Ehrenfeucht-Fra\"{\i}ss\'e characterization of $\mathrm{FO}_{wo=}$ equivalence involves the following condition.  

\begin{definition}

Let $\bf A$ and $\bf B$ be relational structures of some vocabulary $\Sigma$. A relation $p \subseteq A \times B$ is a \emph{partial relativeness correspondence} iff for every $n$-ary relation symbol $R \in \Sigma$ and every $(a_1, b_1), \ldots, (a_n, b_n) \in p$, 
\[ (a_1, \ldots, a_n) \in R^{\bf A} \quad \text{iff} \quad (b_1, \ldots, b_n) \in R^{\bf B}.\]

\end{definition}

\begin{definition}

Let $\bf A$ and $\bf B$ be relational structures of the same vocabulary. $\bf A$ and $\bf B$ are \emph{$m$-finitely relative} via $(I_k)_{k \leq m}$ (denoted $(I_k)_{k \leq m}: {\bf A} \sim_m {\bf B}$) iff the following holds:

\begin{enumerate}[label=\roman{enumi}.]

\item Every $I_k$ is a nonempty set of partial relativeness correspondences.

\item For any $k + 1 \leq m$, any $p \in I_{k+1}$ and any $a \in A$, there are $b \in B$ and $q \in I_k$ such that $q \supseteq p$ and $(a,b) \in q$ (forth condition).

\item For any $k + 1 \leq m$, any $p \in I_{k+1}$ and any $b \in B$, there are $a \in A$ and $q \in I_k$ such that $q \supseteq p$ and $(a,b) \in q$ (back condition).

\end{enumerate}

\end{definition}

The following result is an immediate consequence of \cite[ Prop.~4.5 and Thm.~4.6]{[CDJ96]}.

\begin{theorem}\label{FOwo=Characterization}
Let ${\bf A}$ and ${\bf B}$ be relational structures of some vocabulary $\Sigma$. For every $r \geq 0$, $r$-tuples $\bar{a} = (a_1, \ldots, a_r) \in A^r$ and $\bar{b}=(b_1, \ldots, b_r) \in B^r$, and $p = \{(a_1, b_1), \ldots, (a_r, b_r)\}$, the following holds:
\begin{enumerate}[label=\roman{enumi}.]
\item There is a sequence $(I_k)_{k \leq m}$ such that $(I_k)_{k \leq m}:{\bf A} \sim_m {\bf B}$ and $p \in I_m$ iff for every equality-free formula $\varphi$ of quantifier rank up to $m$ with at most $r$ distinct free variables, ${\bf B} \models \varphi[\bar{b}]$ just in case ${\bf A} \models \varphi[\bar{a}]$.
\item For every $m \geq 0$ there is a sequence $(I_k)_{k \leq m}$ such that $(I_k)_{k \leq m}:{\bf A} \sim_m {\bf B}$ and $p \in I_m$ iff for every equality-free formula $\varphi$ with at most $r$ distinct free variables, ${\bf B} \models \varphi[\bar{b}]$ just in case ${\bf A} \models \varphi[\bar{a}]$.
\end{enumerate}
\end{theorem} 

\new{For $m\in\mathbb N$ and a tuple $\bar a$ over a relational structure ${\bf A}$, we write $\mathit{tp}^m_{\bf A}(\bar a)$ for the set of formulae of quantifier rank at most $m$ in $\mathit{tp}^{\FOwo}_{\bf A}(\bar a)$.}


\begin{newtext}
\subsection{Uniformly Definable Updates and Types of Bounded Rank}\label{subsec:rank}

The second condition of the Bounded Exploration postulate is a cross-state invariance of update sets under $\FOwo$-types of bounded quantifier rank:

\begin{proposition}\label{prop:rank}
Let $W$ cover the updates of $A$. Then $A$ has uniformly definable updates with respect to $W$ iff there is $m\in\mathbb N$ such that for all states ${\bf I},{\bf I}'$ of $A$ with corresponding states of computation ${\bf S},{\bf S}'$, every dynamic function symbol $f$ of arity $r$, and all critical tuples $\bar a=(a_0,\ldots,a_r)$ of ${\bf S}$ and $\bar b=(b_0,\ldots,b_r)$ of ${\bf S}'$ with $\mathit{tp}^m_{{\bf S}|_W}(\bar a)=\mathit{tp}^m_{{\bf S}'|_W}(\bar b)$:
\[
((f,(a_1,\ldots,a_r)),a_0)\in\tau_A({\bf I})-{\bf I}\iff((f,(b_1,\ldots,b_r)),b_0)\in\tau_A({\bf I}')-{\bf I}'.
\]
\end{proposition}

\begin{proof}
($\Rightarrow$) Take $m$ at least the quantifier ranks of all $\psi_f$. ($\Leftarrow$) Since $\Sigma_W$ is finite and relational, for fixed $m$ and $r$ there are, up to logical equivalence, only finitely many equality-free formulae of quantifier rank at most $m$ with free variables among $x_0,\ldots,x_r$ (by induction on $m$). Fix a finite set $F_m$ of representatives. A rank-$m$ type $\theta$ is then determined by $\theta\cap F_m$, and $H_\theta=\bigwedge(\theta\cap F_m)$ satisfies ${\bf B}\models H_\theta[\bar b]$ iff $\mathit{tp}^m_{\bf B}(\bar b)=\theta$. Let $\Theta_f$ be the set of rank-$m$ types $\mathit{tp}^m_{{\bf S}|_W}(\bar a)$ of the tuples $\bar a$ with $((f,(a_1,\ldots,a_r)),a_0)\in\tau_A({\bf I})-{\bf I}$, over all states ${\bf I}$ and corresponding ${\bf S}$, and let $\psi_f=\bigvee_{\theta\in\Theta_f}H_\theta$. If $((f,(a_1,\ldots,a_r)),a_0)\in\tau_A({\bf I})-{\bf I}$, then ${\bf S}|_W\models\psi_f[\bar a]$. Conversely, if ${\bf S}|_W\models\psi_f[\bar a]$ for a critical tuple $\bar a$, then $\mathit{tp}^m_{{\bf S}|_W}(\bar a)$ is the rank-$m$ type of an update tuple in some state, and the hypothesis gives $((f,(a_1,\ldots,a_r)),a_0)\in\tau_A({\bf I})-{\bf I}$.
\end{proof}

The bound on the quantifier rank cannot be dropped: for the algorithms $A_X$ of Remark~\ref{rem:AX}, the $\FOwo$-type of a critical tuple determines the size of the primary part, so that $A_X$ satisfies the condition of Proposition~\ref{prop:rank} without a bound on the rank for every $X$~\cite{[FS-Corr]}.

\subsection{From Formulae over Critical Structures to Terms}\label{subsec:translation}

A \emph{source assignment} is a tuple $\sigma=((i_0,j_0),\ldots,(i_r,j_r))$ with $1\le i_l\le m$ and $0\le j_l\le n_{i_l}$. Let $\bar y_0,\ldots,\bar y_r$ be pairwise disjoint tuples of fresh variables, $\bar y_l$ of length $k_{i_l}$, and put $t^l=t_{i_lj_l}(\bar y_l)$ and $\Phi_l=\varphi_{i_l}(\bar y_l)$. An instantiation $\bar c=(\bar c_0,\ldots,\bar c_r)$ of $(\bar y_0,\ldots,\bar y_r)$ is \emph{$\sigma$-admissible} in ${\bf S}$ if each $\bar c_l$ is witnessing for $\alpha_{i_l}$, and then
\[
\bar a(\bar c)=(\val_{{\bf S},\bar c}(t^0),\ldots,\val_{{\bf S},\bar c}(t^r))
\]
is a critical tuple of ${\bf S}$. By the definition of critical values, every critical tuple of length $r+1$ is of the form $\bar a(\bar c)$ for some $\sigma$ and some $\sigma$-admissible $\bar c$.

\begin{lemma}\label{lem:transl}
For every equality-free formula $\psi(x_0,\ldots,x_r)$ of vocabulary $\Sigma_W$ and every source assignment $\sigma$ there is a Boolean-valued term $t^\sigma_\psi$ of vocabulary $(\Sigma_1,\Sigma_2\cup\Sigma_{\cal K}\cup\{\texttt{atomic}\},\Sigma_F)$, with free variables among $\bar y_0,\ldots,\bar y_r$, such that for every state of computation ${\bf S}$ and every $\sigma$-admissible $\bar c$:
\[
\val_{{\bf S},\bar c}(t^\sigma_\psi)=\texttt{true}\iff{\bf S}|_W\models\psi[\bar a(\bar c)].
\]
\end{lemma}

\begin{proof}
A \emph{reference} is an expression $\langle i,j,\bar y\rangle$ with $0\le j\le n_i$ or $\langle i,\bar y\rangle$, where $\bar y$ is a tuple of $k_i$ variables. Under an instantiation $\bar c$ for which $\bar c(\bar y)$ is witnessing for $\alpha_i$, the reference $\langle i,j,\bar y\rangle$ denotes $\val_{{\bf S},\bar c}(t_{ij}(\bar y))$ and $\langle i,\bar y\rangle$ denotes $\str_i(\bar c(\bar y))$. For a formula $\psi$ and a map $\rho$ from its free variables to references, define the term $T(\psi,\rho)$ by recursion on $\psi$:
\begin{itemize}
\item $T(R_{\alpha_i}(z_0,\ldots,z_{n_i},w),\rho)=\bigwedge_{l=0}^{n_i}\big(t_{il}(\bar y)=t_{i'_lj'_l}(\bar y'_l)\big)$ if $\rho(w)=\langle i,\bar y\rangle$ and $\rho(z_l)=\langle i'_l,j'_l,\bar y'_l\rangle$ for all $l$; otherwise $T(R_{\alpha_i}(z_0,\ldots,z_{n_i},w),\rho)=\texttt{false}$.
\item $T(\neg\psi,\rho)=\neg T(\psi,\rho)$ and $T(\psi\wedge\chi,\rho)=T(\psi,\rho)\wedge T(\chi,\rho)$.
\item $T(\exists z\,\psi,\rho)=\bigvee_{i,j}\exists\bar y\,\big(\varphi_i(\bar y)\wedge T(\psi,\rho[z\mapsto\langle i,j,\bar y\rangle])\big)\vee\bigvee_i\exists\bar y\,\big(\varphi_i(\bar y)\wedge T(\psi,\rho[z\mapsto\langle i,\bar y\rangle])\big)$, with $\bar y$ fresh, $1\le i\le m$ and $0\le j\le n_i$.
\end{itemize}
$T(\psi,\rho)$ is a term (see the remark after Definition~\ref{Def:terms}). We show by induction on $\psi$: for every $\rho$ and every instantiation $\bar c$ under which all references in the range of $\rho$ denote, $\val_{{\bf S},\bar c}(T(\psi,\rho))=\texttt{true}$ iff ${\bf S}|_W\models\psi[e]$, where $e(z)$ is the denotation of $\rho(z)$.

\emph{Atoms.} ${\bf S}|_W\models R_{\alpha_i}[e(z_0),\ldots,e(z_{n_i}),e(w)]$ iff $e(w)$ is the string of a witnessing assignment of $\alpha_i$ and $(e(z_0),\ldots,e(z_{n_i}))$ is its value tuple. Strings are not elements of ${\bf S}$, and a string determines its index (Definition~\ref{def:criticalStructure}). Hence the atom is false unless $\rho(w)=\langle i,\bar y\rangle$ and every $\rho(z_l)$ is of the form $\langle i'_l,j'_l,\bar y'_l\rangle$; in that case the assignment is $\bar c(\bar y)$, and the atom holds iff $\val_{{\bf S},\bar c}(t_{il}(\bar y))=\val_{{\bf S},\bar c}(t_{i'_lj'_l}(\bar y'_l))$ for all $l$. \emph{Connectives} are immediate. \emph{Quantifiers.} The domain of ${\bf S}|_W$ consists exactly of the denotations of the references $\langle i,j,\bar y\rangle$ and $\langle i,\bar y\rangle$ under instantiations of $\bar y$ by witnessing assignments of $\alpha_i$; apply the induction hypothesis to $\rho[z\mapsto\cdot]$ and the extended instantiation.

Finally let $t^\sigma_\psi=T(\psi,\rho_\sigma)$ with $\rho_\sigma(x_l)=\langle i_l,j_l,\bar y_l\rangle$, so that $e(x_l)$ is the $l$-th component of $\bar a(\bar c)$.
\end{proof}

\subsection{Characterization}\label{subsec:characterization}

\begin{theorem}\label{thm:characterization}
For every parallel algorithm $A$ there is a behaviourally equivalent parallel ASM with the same vocabulary and background class as $A$.
\end{theorem}

\begin{proof}
For every dynamic function symbol $f$ of arity $r$ and every source assignment $\sigma$ let
\[
\rho_{f,\sigma}=\textbf{forall } \bar y_0,\ldots,\bar y_r \textbf{ with } \Phi_0\wedge\cdots\wedge\Phi_r\wedge t^\sigma_{\psi_f} \textbf{ do } f(t^1,\ldots,t^r):=t^0 \textbf{ enddo},
\]
where guards are padded with conjuncts $y=y$, so that every variable of $\bar y_0,\ldots,\bar y_r$ is free in the guard, as required by Definition~\ref{ASMrules}. Let $r_{\mathcal M}$ be the parallel composition of the finitely many rules $\rho_{f,\sigma}$. By Lemma~\ref{lem:transl}, $r_{\mathcal M}$ is a closed rule of vocabulary $(\Sigma_1,\Sigma_2\cup\Sigma_{\cal K}\cup\{\texttt{atomic}\},\Sigma_F)$. We show that $\Delta(r_{\mathcal M},{\bf S})=\tau_A({\bf I})-{\bf I}$ for every state ${\bf I}$ and every corresponding ${\bf S}$. Since this update set is consistent, $\mathcal{M}=(\setofstates{S}{A},\setofstates{I}{A},r_{\mathcal M})$ is then a parallel ASM (Definition~\ref{ASM}; closure under isomorphisms holds by the Abstract State postulate) with $\tau_{\mathcal M}({\bf I})={\bf I}+(\tau_A({\bf I})-{\bf I})=\tau_A({\bf I})$, hence behaviourally equivalent to $A$.

($\subseteq$) Let $\bar c$ satisfy the guard of $\rho_{f,\sigma}$ in ${\bf S}$. Then $\bar c$ is $\sigma$-admissible and, by Lemma~\ref{lem:transl}, ${\bf S}|_W\models\psi_f[\bar a(\bar c)]$ for the critical tuple $\bar a(\bar c)$. By Definition~\ref{def:udu}, the update produced by this instance, $((f,(a_1,\ldots,a_r)),a_0)$ for $\bar a=\bar a(\bar c)$, belongs to $\tau_A({\bf I})-{\bf I}$; since $A$ does not import, it satisfies the side conditions of Definition~\ref{ASMrules}, so that the instance is defined. ($\supseteq$) Let $((f,(a_1,\ldots,a_r)),a_0)\in\tau_A({\bf I})-{\bf I}$. Since $W$ covers the updates, $\bar a=\bar a(\bar c)$ for some $\sigma$ and some $\sigma$-admissible $\bar c$. By Definition~\ref{def:udu} and Lemma~\ref{lem:transl}, $\bar c$ satisfies the guard of $\rho_{f,\sigma}$, which then produces the update.
\end{proof}

\begin{corollary}\label{cor:characterization}
An algorithm that satisfies the Sequential Time, Abstract State and Background postulates satisfies the Bounded Exploration postulate if and only if it is behaviourally equivalent to a parallel ASM with the same vocabulary and background class.
\end{corollary}

\begin{proof}
By Theorems~\ref{thm:characterization} and~\ref{Th:Plausibility}, since behaviourally equivalent algorithms have the same states and the same one-step transformation.
\end{proof}
\end{newtext}

\section{Conclusions}

In this article we revisited the problem of the ``parallel ASM thesis'' (see \cite{[BG03]}), i.e. to provide a machine-independent definition of parallel algorithm and a proof that these algorithms are faithfully captured by Abstract State Machines. The main motivation is the often uttered conviction that although the mathematical proof is correct, the definition of parallel algorithm given by Blass and Gurevich in \cite{[BG03]} is not convincing, as the postulates reside too much on the technical side and do not provide the same level of intuitive clarity as the postulates for sequential algorithms. Our intention was thus to prove the conjecture in \cite{[SW2012a]}, according to which four simplified postulates suffice to justify ASMs as a general model for parallel computation, i.e. to provide a more intuitive set of postulates and to formally prove that parallel algorithms as stipulated by these new postulates are indeed captured by ASMs.

As a matter of fact, postulates are always debatable, so we open the debate, whether the goal to provide an intuitively clear and acceptable characterisation of synchronous parallel algorithms has now been reached. Technically, the new set of postulates is equivalent to the one given by Blass and Gurevich, as both are captured exactly by ASMs.

The set of postulates for synchronous parallel algorithms presented in this article is rather close to the one used for sequential algorithms \cite{[Gurevich00]}, which has been widely accepted by the scientific community. There are two main differences. The first one is the addition of a background postulate analogous to the background postulate in \cite{[BG03]}. In a sense, this postulates makes all assumptions about the background of a computation explicit. It is only necessary, as there is a need to exploit tuples and multisets, which are not required in sequential algorithms. In a strict formal sense there is also a background for sequential algorithms, but the assumptions have been left implicit. The second one is the extension of the bounded exploration postulates, which still claims a finite set of exploration witness terms that determine update sets, but instead of simple ground terms now multiset comprehension terms are needed. \new{In addition, the update sets must be uniformly definable in the finite critical structures given by the values of these terms. In the sequential case this uniformity is a consequence of bounded exploration; in the parallel case it has to be required explicitly, since otherwise the postulates admit objects, such as $A_X$ in Remark~\ref{rem:AX}, that are not behaviourally equivalent to any parallel ASM over their own background.} By means of these the varying parallel branches in a parallel computation that depend not only on the algorithm but also on the state are captured and there is no need for a separate concept of ``proclet''. 

With the new parallel ASM thesis at hand we will now proceed further towards a theis for concurrent ASMs capturing asynchronous parallel algorithms. The work in \cite{boerger:ai2015} contains a first attempt in this direction, which so far is restricted to families of sequential algorithms. We believe that with the result in this article we can achieve an easy generalisation to families of parallel algorithms.

\bibliographystyle{elsarticle-num} 
\bibliography{thesis}

\end{document}